\documentclass{llncs}

\usepackage{microtype}
\usepackage{enumerate}
\usepackage{boxedminipage}
\usepackage{tikz,tkz-graph}
\usepackage{hyperref}
\usepackage{amssymb,amsmath}
\spnewtheorem{oproblem}{Open Problem}{\bfseries}{\itshape}
\spnewtheorem{observation}{Observation}{\bfseries}{\itshape}

\newcounter{ctrclaim}[theorem]
\newcounter{ctrcase}[theorem]
\newcounter{ctrsubcase}[ctrcase]

\newcommand\displaycase[1]{{\bfseries#1}}

\newcommand{\clm}[1]{\medskip\phantomsection\refstepcounter{ctrclaim}\noindent\displaycase{Claim \thectrclaim. }{\itshape #1}\\}
\newcommand{\thmcase}[1]{\medskip\phantomsection\refstepcounter{ctrcase}\noindent\displaycase{Case \thectrcase. }{\itshape #1}\\}

\newcommand{\GI}{{\sf GI}}
\newcommand{\NP}{{\sf NP}}
\newcommand{\FPT}{{\sf FPT}}
\newcommand{\XP}{{\sf XP}}
\newcommand{\dia}{\hfill{$\diamond$}}
\DeclareMathOperator{\claw}{claw}
\DeclareMathOperator{\diamondgraph}{diamond}
\DeclareMathOperator{\gem}{gem}
\DeclareMathOperator{\paw}{paw}
\DeclareMathOperator{\crossedhouse}{crossed\ house}
\DeclareMathOperator{\cw}{cw}

\newcommand{\problemdef}[3]{
        \begin{center}
                \begin{boxedminipage}{.99\textwidth}
                        \textsc{{#1}}\\[2pt]
                        \begin{tabular}{ r p{0.8\textwidth}}
                                \textit{~~~~Instance:} & {#2}\\
                                \textit{Question:} & {#3}
                        \end{tabular}
                \end{boxedminipage}
        \end{center}
}

\newcommand{\ssi}{\subseteq_i}
\newcommand{\si}{\supseteq_i}

\title{Graph Isomorphism for $(H_1,H_2)$-Free Graphs: An~Almost Complete Dichotomy\thanks{Research supported by the London Mathematical Society (SC7-1718-04), ANR project HOSIGRA (ANR-17-CE40-0022), EPSRC (EP/K025090/1) and the Leverhulme Trust (RPG-2016-258).
An extended abstract of this paper appeared in the proceedings of WADS 2019~\cite{BDJP18}.
}}

\author{Marthe~Bonamy\inst{1}
\and
Nicolas~Bousquet\inst{2}
\and
Konrad~K.~Dabrowski\inst{3}
\and
Matthew~Johnson\inst{3}
\and
Dani\"el~Paulusma\inst{3}
\and
Th\'eo~Pierron\inst{4}
}

\institute{CNRS, LaBRI, Université de Bordeaux, France
\email{marthe.bonamy@u-bordeaux.fr}\and
Univ. Grenoble Alpes, CNRS, Grenoble INP, G-SCOP, France
\email{nicolas.bousquet@grenoble-inp.fr}\and
Department of Computer Science, Durham University, UK
\email{\{konrad.dabrowski,matthew.johnson2,daniel.paulusma\}@durham.ac.uk}\and
Université de Bordeaux, Bordeaux INP, CNRS, LaBRI, UMR5800, F-33400 Talence, France
\email{tpierron@labri.fr}}

\begin{document}
\maketitle

\begin{abstract}
We resolve the computational complexity of {\sc Graph Isomorphism} for classes of graphs characterized by two forbidden induced subgraphs~$H_1$ and~$H_2$ for all but six pairs $(H_1,H_2)$.
Schweitzer had previously shown that the number of open cases was finite, but without specifying the open cases. 
Grohe and Schweitzer proved that {\sc Graph Isomorphism} is polynomial-time solvable on graph classes of bounded clique-width. 
Our work combines known results such as these with new results. 
By exploiting a relationship between {\sc Graph Isomorphism} and clique-width, we simultaneously reduce the number of open cases for boundedness of clique-width for $(H_1,H_2)$-free graphs to five.

\end{abstract}

\section{Introduction}

The {\sc Graph Isomorphism} problem, which is that of deciding whether two given graphs are isomorphic, is a central problem in Computer Science.
It is not known whether {\sc Graph Isomorphism} is polynomial-time solvable.
However, it is not \NP-complete unless the polynomial hierarchy collapses~\cite{Sc88}.
Analogously to the use of the notion of \NP-completeness, we can say that a problem is {\sc Graph Isomorphism-}complete (abbreviated to \GI-complete).
Babai~\cite{Ba16} proved that {\sc Graph Isomorphism} can be solved in quasi-polynomial time.

In order to increase understanding of the computational complexity of {\sc Graph Isomorphism}, it is natural to place restrictions on the input.
This approach has established that on many graph classes {\sc Graph Isomorphism} is polynomial-time solvable, but that on many others the problem remains \GI-complete.
We refer to~\cite{isgci} for a survey, but some recent examples include a polynomial-time algorithm for unit square graphs~\cite{Ne16}, a complexity dichotomy for $H$-induced-minor-free graphs~\cite{BOS15} and a polynomial-time algorithm for graphs of bounded maximum degree~\cite{GNS18} (improving on the runtime of previous polynomial-time algorithms on graphs of bounded maximum degree~\cite{BKL83,Lu82}).

In this paper we consider the {\sc Graph Isomorphism} problem for hereditary graph classes, which are the classes of graphs that are closed under vertex deletion.
It is readily seen that a graph class~${\cal G}$ is hereditary if and only if there exists a family of graphs~${\cal F}_{\cal G}$, such that the following holds: a graph~$G$ belongs to~${\cal G}$ if and only if~$G$ does not contain any graph from~${\cal F}_{\cal G}$ as an induced subgraph.
We implicitly assume that~${\cal F}_{\cal G}$ is a family of minimal forbidden induced subgraphs, in which case~${\cal F}_{\cal G}$ is unique.
We note that~${\cal F}_{\cal G}$ may have infinite size.
For instance, if~${\cal G}$ is the class of bipartite graphs, then~${\cal F}_{\cal G}$ consists of all odd cycles.

A natural direction for a {\it systematic} study of the computational complexity of {\sc Graph Isomorphism} is to consider graph classes~${\cal G}$, for which~${\cal F}_{\cal G}$ is small, starting with the case where~${\cal F}_{\cal G}$ has size~$1$.
A graph is {\em $H$-free} if it does not contain~$H$ as induced subgraph; conversely, we write $H \ssi G$ to denote that~$H$ is an induced subgraph of~$G$.
The classification for $H$-free graphs can be found in a technical report of Booth and Colbourn~\cite{BC79}, who credit the result to an unpublished manuscript of Colbourn and Colbourn; another proof of it appears in a paper of Kratsch and Schweitzer~\cite{KS12}.

\begin{theorem}[see~\cite{BC79,KS12}]\label{thm:gi-one-graph}
Let~$H$ be a graph. Then {\sc Graph Isomorphism} on $H$-free graphs is polynomial-time solvable if $H \subseteq_i P_4$ and \GI-complete otherwise.
\end{theorem}
Later, Colbourn~\cite{Co81} proved that {\sc Graph Isomorphism} is polynomial-time solvable even for the class of permutation graphs, which form a superclass of the class of $P_4$-free graphs.
Classifying the case where~${\cal F}_{\cal G}$ has size~$2$ is much more difficult than the size-$1$ case.
Kratsch and Schweitzer~\cite{KS12} initiated this classification.
Schweitzer~\cite{Sc17} later extended the results of~\cite{KS12} and proved that only a finite number of cases remain open.
This leads to our research question:

\begin{quote}
{\em Is it possible to determine the computational complexity of {\sc Graph Isomorphism} for $(H_1,H_2)$-free graphs\footnote{A graph is {\em $(H_1,H_2)$-free} if it has no induced subgraph isomorphic to~$H_1$ or~$H_2$.} for all pairs $H_1,H_2$?}
\end{quote}

\noindent
The analogous research question for $H$-induced-minor-free graphs was fully answered by Belmonte, Otachi and Schweitzer~\cite{BOS15}, who also determined all graphs~$H$ for which the class of $H$-induced-minor-free graphs has bounded clique-width.
Similar classifications for {\sc Graph Isomorphism}~\cite{Po88} and boundedness of clique-width~\cite{DP16} are also known for $H$-minor-free graphs.

Lokshtanov et al.~\cite{LPPS17} recently gave an \FPT\ algorithm for {\sc Graph Isomorphism} with parameter~$k$ on graph classes of treewidth at most~$k$, and this has since been improved by Grohe et al.~\cite{GNSW18}.
Whether an \FPT\ algorithm exists when parameterized by clique-width is still open.
Grohe and Schweitzer~\cite{GS15} proved membership of \XP.

\begin{theorem}[\cite{GS15}]\label{thm:gi-poly-bdd-cw}
For every constant~$c$, {\sc Graph Isomorphism} is polynomial-time solvable on graphs of clique-width at most~$c$.
\end{theorem}
Grohe and Neuen~\cite{GN19} have since improved this result by showing that the more general {\sc Canonisation} problem is also in \XP\ when parameterized by clique-width.

\subsection*{Our Results}

By combining known results with Theorem~\ref{thm:gi-poly-bdd-cw} we narrow the list of open cases for {\sc Graph Isomorphism} on $(H_1,H_2)$-free graphs to~14.
Of these~14 cases, we prove that three of them are polynomial-time solvable (Section~\ref{sec:poly}) and five others are \GI-complete (Section~\ref{sec:GI-complete}).
Thus we reduce the number of open cases to six.

Besides Theorem~\ref{thm:gi-poly-bdd-cw}, there is another reason why results for clique-width are of importance for {\sc Graph Isomorphism}.
Namely, Schweitzer~\cite{Sc17} pointed out great similarities between proving unboundedness of clique-width of some graph class~${\cal G}$ and proving that {\sc Graph Isomorphism} stays \GI-complete for~${\cal G}$.
We will illustrate these similarities by showing that our construction demonstrating that {\sc Graph Isomorphism} is \GI-complete for $(\gem,P_1+\nobreak 2P_2)$-free graphs can also be used to show that this class has unbounded clique-width.
This reduces the number of pairs $(H_1,H_2)$ for which we do not know if the class of $(H_1,H_2)$-free graphs has bounded clique-width from six~\cite{DLP17} to five.
As such, our paper also continues a project~\cite{BDJLPZ17,BDHP17,DDP17,DHP0,DLP17,DP16} aiming to classify the boundedness of clique-width of $(H_1,H_2)$-free graphs for all pairs $(H_1,H_2)$; see Section~\ref{sec:cw} (or a recent survey on clique-width~\cite{DJP19}) for an overview of the known and open cases.

In Section~\ref{sec:gi-classification} we present our main theorem, which states exactly for which classes of $(H_1,H_2)$-free graphs {\sc Graph Isomorphism} is known to be polynomial-time solvable, for which it is \GI-complete and for which six cases the complexity remains open.

\section{Preliminaries}\label{sec:prelim}
We consider only finite, undirected graphs without multiple edges or self-loops.
An {\em isomorphism} from a graph~$G$ to a graph~$H$ is a bijection $f:V(G) \to V(H)$ such that $vw \in E(G)$ if and only if $f(v)f(w) \in E(H)$.
For a function $f:X \to Y$, if $X' \subseteq X$, we define $f(X'):=\{f(x) \in Y \;|\; x \in X'\}$.
The {\sc Graph Isomorphism} problem is defined as follows.
\problemdef{{\sc Graph Isomorphism}}{Graphs~$G$ and~$H$.}{Is there an isomorphism from~$G$ to~$H$?}

The {\em disjoint union} $(V(G)\cup V(H), E(G)\cup E(H))$ of two vertex-disjoint graphs~$G$ and~$H$ is denoted by~$G+\nobreak H$ and the disjoint union of~$r$ copies of a graph~$G$ is denoted by~$rG$.
For a subset $S\subseteq V(G)$, we let~$G[S]$ denote the subgraph of~$G$ {\it induced} by~$S$, which has vertex set~$S$ and edge set $\{uv\; |\; u,v\in S, uv\in E(G)\}$.
If $S=\{s_1,\ldots,s_r\}$, then we may write $G[s_1,\ldots,s_r]$ instead of $G[\{s_1,\ldots,s_r\}]$.
Recall that for two graphs~$G$ and~$G'$ we write $G'\ssi G$ to denote that~$G'$ is an induced subgraph of~$G$.
For a set of graphs $\{H_1,\ldots,H_p\}$, a graph~$G$ is {\em $(H_1,\ldots,H_p)$-free} if it has no induced subgraph isomorphic to a graph in $\{H_1,\ldots,H_p\}$;
recall that if~$p=1$, we may write $H_1$-free instead of $(H_1)$-free.

Let $G$ be a graph.
The set $N(u)=\{v\in V\; |\; uv\in E\}$ denotes the {\em (open) neighbourhood} of $u\in V(G)$ and $N[u]=N(u) \cup \{u\}$ denotes the {\em closed neighbourhood} of~$u$.
The {\em degree}~$d_G(v)$ of a vertex~$v$ in a graph~$G$ is the number of vertices in~$G$ that are adjacent to~$v$.
A vertex~$v\in V(G)$ is {\em dominating} if every vertex in $V(G) \setminus \{v\}$ is adjacent to~$v$.
If~$X$ is a set of vertices in~$G$, then~$X$ is {\em dominating} if every vertex in $V(G) \setminus X$ has a neighbour in~$X$.
A vertex and an edge are {\em incident} if the vertex is one of the two end-vertices of the edge.
A {\em (connected) component} of~$G$ is a maximal subset of vertices that induces a connected subgraph of~$G$; it  is {\em non-trivial} if it has at least two vertices, otherwise it is {\em trivial}.
The {\em complement}~$\overline{G}$ of a graph~$G$ has vertex set $V(\overline{G})=V(G)$ such that two vertices are adjacent in~$\overline{G}$ if and only if they are not adjacent in~$G$.

The graphs $C_t$, $K_t$, $K_{1,t-1}$ and~$P_r$ denote the cycle, complete graph, star and path on~$t$ vertices, respectively.
Let~$K_{1,t}^+$ and~$K_{1,t}^{++}$ be the graphs obtained from~$K_{1,t}$ by subdividing one edge once or twice, respectively.
The graphs $K_{1,3}$, $\overline{2P_1+P_2}$, $\overline{P_1+P_3}$, $\overline{P_1+P_4}$ and~$\overline{2P_1+P_3}$ are also called the {\em $\claw$}, {\em $\diamondgraph$}, {\em $\paw$}, {\em $\gem$} and {\em $\crossedhouse$}, respectively.
The graph~$S_{h,i,j}$, for $1\leq h\leq i\leq j$, denotes the {\em subdivided claw}, that is, the tree that has only one vertex~$x$ of degree~$3$ and exactly three leaves, which are at distance~$h$,~$i$ and~$j$ from~$x$, respectively.
Observe that $S_{1,1,1}=K_{1,3}$.
We use~${\cal S}$ to denote the set of graphs every component of which is either a subdivided claw or a path on at least one vertex.
A {\em subdivided star} is a graph obtained from a star by subdividing its edges an arbitrary number of times.
A graph is a {\em path star forest} if all of its connected components are subdivided stars.
A graph is a {\em linear forest} if every component of~$G$ is a path (on at least one vertex).

We will need the following results.
\begin{lemma}[\cite{Sc17}]\label{lem:K_t-K_1t-free-gi-poly}
For every fixed~$t$, {\sc Graph Isomorphism} is polynomial-time solvable on $(2K_{1,t},K_t)$-free graphs.
\end{lemma}

\begin{lemma}[\cite{Sc17}]\label{lem:K_t-P_5-free-gi-poly}
For every fixed~$t$, {\sc Graph Isomorphism} is polynomial-time solvable on $(K_t,P_5)$-free graphs.
\end{lemma}

Let~$G$ be a graph and let $X,Y \subseteq V(G)$ be disjoint sets.
The edges between~$X$ and~$Y$ form a {\em perfect matching} if every vertex in~$X$ is adjacent to exactly one vertex in~$Y$ and vice versa.
A vertex $x \in V(G) \setminus Y$ is {\em complete} (resp. {\em anti-complete}) to~$Y$ if it is adjacent (resp. non-adjacent) to every vertex in~$Y$.
Similarly, $X$ is complete (resp. anti-complete) to~$Y$ if every vertex in~$X$ is complete (resp. anti-complete) to~$Y$.
A graph is {\em bipartite} if its vertex set can be partitioned into two (possibly empty) independent sets.
A graph is {\em split} if its vertex set can be partitioned into a clique and an independent set.
A graph is {\em complete multipartite} if its vertex set can be partitioned into independent sets $V_1,\ldots,V_k$ such that~$V_i$ is complete to~$V_j$ whenever $i \neq j$; if $k=2$, then the graph is {\em complete bipartite}.
We will need the following result.

\begin{lemma}[\cite{Olariu88}]\label{lem:olariu}
Every connected $(\overline{P_1+P_3})$-free graph is either complete multipartite or $K_3$-free.
\end{lemma}

\subsection{Clique-width}\label{sec:clique-width}
The {\em clique-width} of a graph~$G$, denoted by~$\cw(G)$, is the minimum number of labels needed to construct~$G$ using the following four operations:
\begin{enumerate}[(i)]
\item create a new graph consisting of a single vertex~$v$ with label~$i$;
\item take the disjoint union of two labelled graphs~$G_1$ and~$G_2$;
\item join each vertex with label~$i$ to each vertex with label~$j$ ($i\neq j$);
\item rename label~$i$ to~$j$.
\end{enumerate}

\noindent
A class of graphs~${\cal G}$ has bounded clique-width if there is a constant~$c$ such that the clique-width of every graph in~${\cal G}$ is at most~$c$; otherwise the clique-width of~${\cal G}$ is {\em unbounded}.

Let~$G$ be a graph.
We define the following operations.
For an induced subgraph $G'\ssi G$, the {\em subgraph complementation} operation (acting on~$G$ with respect to~$G'$) replaces every edge present in~$G'$ by a non-edge, and vice versa,
that is, the resulting graph has vertex set~$V(G)$ and edge set $(E(G) \setminus E(G')) \cup \{xy\;|\; x,y \in V(G'), x \neq y, xy \notin E(G')\}$.
Similarly, for two disjoint vertex subsets~$S$ and~$T$ in~$G$, the {\em bipartite complementation} operation with respect to~$S$ and~$T$ acts on~$G$ by replacing every edge with one end-vertex in~$S$ and the other in~$T$ by a non-edge and vice versa.

We now state some useful facts about how these two operations (and some others) influence the clique-width of a graph.
We will use these facts throughout the paper.
Let $k\geq 0$ be a constant and let~$\gamma$ be some graph operation.
We say that a graph class~${\cal G'}$ is {\em $(k,\gamma)$-obtained} from a graph class~${\cal G}$ if the following two conditions hold:
\begin{enumerate}[(i)]
\item every graph in~${\cal G'}$ is obtained from a graph in~${\cal G}$ by performing~$\gamma$ at most~$k$ times, and
\item for every $G\in {\cal G}$ there exists at least one graph in~${\cal G'}$ obtained from~$G$ by performing~$\gamma$ at most~$k$ times.
\end{enumerate}

\noindent
We say that~$\gamma$ {\em preserves} boundedness of clique-width if for any finite constant~$k$ and any graph class~${\cal G}$, any graph class~${\cal G}'$ that is $(k,\gamma)$-obtained from~${\cal G}$ has bounded clique-width if and only if~${\cal G}$ has bounded clique-width.

\begin{enumerate}[\bf F{a}ct 1.]
\item \label{fact:del-vert}Vertex deletion preserves boundedness of clique-width~\cite{LR04}.
\item \label{fact:comp}Subgraph complementation preserves boundedness of clique-width~\cite{KLM09}.
\item \label{fact:bip}Bipartite complementation preserves boundedness of clique-width~\cite{KLM09}.
\end{enumerate}

\noindent
We need the following two lemmas on clique-width.

\begin{lemma}[\cite{BDHP16}]\label{lem:2P_1P3-split-bdd-cw}
The class of $\overline{2P_1+P_3}$-free split graphs has bounded clique-width.
\end{lemma}

\begin{lemma}[\cite{BKM06}]\label{lem:K3-P6-bdd-cw}
The class of $(K_3,P_6)$-free graphs has bounded clique-width.
\end{lemma}

Since complete multipartite graphs have clique-width at most~$2$, and the clique-width of a graph is equal to the maximum clique-width of its components, we can use Lemma~\ref{lem:olariu} to extend Lemma~\ref{lem:K3-P6-bdd-cw} into the following (previously-known) corollary.

\begin{corollary}\label{cor:pawP5-free-bdd-cw}
The class of $(\paw,P_6)$-free graphs has bounded clique-width.
\end{corollary}

We also need the special case of~\cite[Theorem~3]{DP16} when $V_{0,i}=V_{i,0}=\emptyset$ for $i \in \{1,\ldots,n\}$.

\begin{lemma}[\cite{DP16}]\label{lem:generalunbounded}
For $m\geq 1$ and $n >\nobreak m+\nobreak 1$ the clique-width of a graph~$G$ is at least 
$\lfloor\frac{n-1}{m+1}\rfloor+\nobreak 1$
if~$V(G)$ has a partition into sets $V_{i,j}\; (i,j \in \{1,\ldots,n\})$ with 
the following properties:
\begin{enumerate}
\item \label{prop:v_ij-nonempty} $|V_{i,j}|\geq 1$ for all $i,j\geq 1$.
\item \label{prop:row-connected} $G[\cup^n_{j=1}V_{i,j}]$ is connected for all $i\geq 1$.
\item \label{prop:column-connected} $G[\cup^n_{i=1}V_{i,j}]$ is connected for all $j\geq 1$.
\item \label{prop:v_ij-nbrs} For $i,j,k,\ell\geq 1$, if a vertex of~$V_{i,j}$ is adjacent to a vertex of~$V_{k,\ell}$, then $|k-i|\leq m$ and $|\ell-j| \leq m$.
\end{enumerate}
\end{lemma}

\section{New Polynomial-Time Results}\label{sec:poly}

\begin{figure}
\begin{center}
\begin{tabular}{ccc}
\begin{minipage}{0.25\textwidth}
\begin{center}
\scalebox{0.9}{
\begin{tikzpicture}[scale=1,rotate=90]
\GraphInit[vstyle=Simple]
\SetVertexSimple[MinSize=6pt]
\Vertices{circle}{a,b,c,d,e}
\Edges(e,a,b)
\Edges(b,c,d,e,b)
\Edges(b,d)
\Edges(c,e)
\end{tikzpicture}}
\end{center}
\end{minipage}
&
\begin{minipage}{0.25\textwidth}
\begin{center}
\scalebox{0.9}{
\begin{tikzpicture}[scale=1,rotate=90]
\GraphInit[vstyle=Simple]
\SetVertexSimple[MinSize=6pt]
\Vertices{circle}{a,b,c,d,e}
\Edges(d,e,a,b,c)
\end{tikzpicture}}
\end{center}
\end{minipage}
&
\begin{minipage}{0.25\textwidth}
\begin{center}
\scalebox{0.9}{
\begin{tikzpicture}[scale=1,rotate=90]
\GraphInit[vstyle=Simple]
\SetVertexSimple[MinSize=6pt]
\Vertices{circle}{a,b,c,d,e}
\Edges(e,a,b)
\Edges(c,d)
\end{tikzpicture}}
\end{center}
\end{minipage}
\\
\\
$\overline{2P_1+P_3}$ & $P_5$ & $P_2+\nobreak P_3$ 
\end{tabular}
\end{center}
\caption{\label{fig:forb-easy}Forbidden induced subgraphs from Theorems~\ref{thm:coP5-2P1+P3-easy} and~\ref{thm:co2P1+P3-2P1+P3-easy}.}
\end{figure}
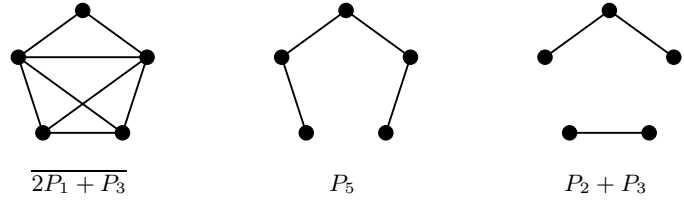

\begin{sloppypar}
In this section we prove Theorems~\ref{thm:coP5-2P1+P3-easy} and~\ref{thm:co2P1+P3-2P1+P3-easy}, which states that {\sc Graph Isomorphism} is polynomial-time solvable on $(\overline{2P_1+P_3},P_5)$-free graphs and $(\overline{2P_1+P_3},P_2+\nobreak P_3)$-free graphs, respectively (see also \figurename~\ref{fig:forb-easy}).
The complexity of {\sc Graph Isomorphism} on $(\overline{2P_1+P_3},\allowbreak 2P_2)$-free graphs was previously unknown, but since this class is contained in the classes of $(\overline{2P_1+P_3},P_5)$-free graphs and $(\overline{2P_1+P_3},P_2+\nobreak P_3)$-free graphs, Theorems~\ref{thm:coP5-2P1+P3-easy} and~\ref{thm:co2P1+P3-2P1+P3-easy} both imply that {\sc Graph Isomorphism} is also polynomial-time solvable on this class.
\end{sloppypar}

Before proving Theorems~\ref{thm:coP5-2P1+P3-easy} and~\ref{thm:co2P1+P3-2P1+P3-easy}, we first prove a useful lemma (see also \figurename~\ref{fig:special-lemma}).

\begin{figure}
\begin{center}
\begin{tikzpicture}[every node/.style={inner sep=1.4pt, outer sep=0pt, circle, draw,fill=black}]
    \node (A) at (90:2){};
    \node (B) at (18:2){};
    \node (C) at (306:2){};
    \node (D) at (234:2){};
    \node (E) at (162:2){};

    \draw[ultra thick] (A) -- (B) -- (C) -- (D) -- (E) -- (A) -- (C) -- (E) -- (B) -- (D) -- (A);
    
    \begin{scope}[rotate=90,shift={(2,0)}]  
      \draw[fill=white] (0,0) ellipse (0.25cm and 0.75cm);
      \node[fill=white] (A1) at (0,0) {};
      \node (A2) at (0,-.5) {};
      \node (A3) at (0,.5) {};

      \node[fill=none,draw=none] (A) at (0,-1) {$A_1^G$};
      \node (NA1) at (1,-.25) {};
      \node (NA2) at (1,.25) {};
      \draw (1,0) ellipse (0.25cm and 0.5cm);
      \node[fill=none,draw=none] at (1,-1) {$N_1^G$};
    \end{scope}

    \begin{scope}[rotate=18,shift={(2,0)}]
      \draw[fill=white] (0,0) ellipse (0.25cm and 0.5cm);
      \node[fill=white] (B1) at (0,.25) {};
      \node (B2) at (0,-.25) {};
      \node (NB1) at (1,-.5) {};
      \node (NB2) at (1,.5) {};
      \node (NB3) at (1,0) {};
      \draw (1,0) ellipse (0.25cm and 0.75cm);
      \node[fill=none,draw=none] (B) at (0,-0.75) {$A_2^G$};
      \node[fill=none,draw=none] at (1,-1) {$N_2^G$}; 
    \end{scope}

    \begin{scope}[rotate=306,shift={(2,0)}]
      \draw[fill=white] (0,0) ellipse (0.25cm and 0.75cm);
      \node[fill=white] (C1) at (0,0) {};
      \node (C2) at (0,-.5) {};
      \node (C3) at (0,.5) {};
      \node (NC1) at (1,-.5) {};
      \node (NC2) at (1,.5) {};
      \node (NC3) at (1,0) {};
      \draw (1,0) ellipse (0.25cm and 0.75cm);
      \node[fill=none,draw=none] (C) at (0,-1.1) {$A_3^G$};
      \node[fill=none,draw=none] at (1,-1.1) {$N_3^G$};
    \end{scope}

    \begin{scope}[rotate=234,shift={(2,0)}]
      \draw[fill=white] (0,0) ellipse (0.25cm and 0.5cm);
      \node[fill=white] (D1) at (0,.25) {};
      \node (D2) at (0,-.25) {};
      \node (ND) at (1,0) {};
      \draw (1,0) ellipse (0.25cm and 0.25cm);
      \node[fill=none,draw=none] (D) at (0,-.75) {$A_4^G$};
      \node[fill=none,draw=none] at (1,-.6) {$N_4^G$};
    \end{scope}

    \begin{scope}[rotate=162,shift={(2,0)}]
      \draw[fill=white] (0,0) ellipse (0.25cm and 0.5cm);
      \node[fill=white] (E1) at (0,.25) {};
      \node (E2) at (0,-.25) {};
      \node (NE1) at (1,-.25) {};
      \node (NE2) at (1,.25) {};
      \draw (1,0) ellipse (0.25cm and 0.5cm);
      \node[fill=none,draw=none] (E) at (0,-.75) {$A_5^G$};
      \node[fill=none,draw=none] at (1,-.75) {$N_5^G$};
    \end{scope}
    \draw (A2) -- (NA1) -- (A1) -- (NA2);
    \draw (NB3) -- (B1) -- (NB2);
    \draw (B2) -- (NB1);
    \draw (NC1) -- (C2) -- (NC2) -- (C1) -- (NC3);
    \draw (NC2) -- (C3);
    \draw (ND) -- (D1);
    \draw (NE1) -- (E1) -- (NE2) -- (E2) -- (NE1);

    \draw (NA1) -- (NA2);
    \draw (NB1) -- (NB3) -- (NB2);
    \draw (NC1) -- (NC3);

    \draw (NA1) -- (NB2);
    \draw (NA2) -- (NE1);
    \draw (ND) -- (NC1);

    \begin{scope}[shift={(7,-0.5)}]
      \draw (0,0) circle (1cm);
      \node (b1) at (90:0.75){};
      \node (b2) at (18:0.75){};
      \node (b3) at (306:0.75){};
      \node (b4) at (234:0.75){};
      \node (b5) at (162:0.75){};
      \node[draw=none, fill=none] at (0,1.5) {$B^G$};
    \end{scope}

    \draw (b1) -- (b2) -- (b3) -- (b1);
    \draw (b4) -- (b5);

    \draw (NA1) -- (b1);
    \draw (b5) -- (NA1);
    \draw (b5) -- (NB1);
    \draw (NC1) -- (b4) -- (NC2);
  \end{tikzpicture}
\end{center}
\caption{\label{fig:special-lemma}An example of Lemma~\ref{lem:K5-extension} applied to a $\overline{2P_1+P_3}$-free graph.
White vertices denote the vertices of~$K^G$ and thick edges between two sets of vertices indicate that these sets are complete to each other.}
\end{figure}
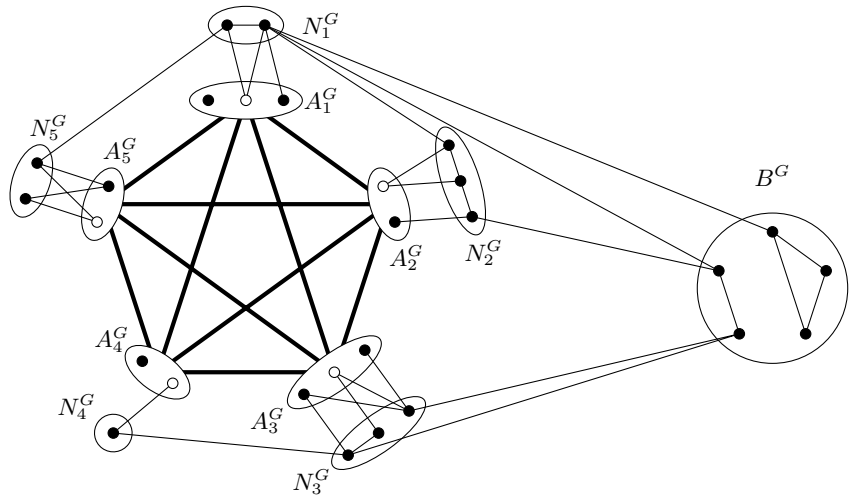

\begin{lemma}\label{lem:K5-extension}
Let~$G$ be a $\overline{2P_1+P_3}$-free graph containing an induced~$K_5$ with vertex set~$K^G$.
Then~$V(G)$ can be partitioned into sets $A^G_1,\ldots,A^G_p,N^G_1,\ldots,N^G_p,B^G$ for some $p \geq 5$ such that:
\begin{enumerate}[(i)]
\item $K^G \subseteq \bigcup A^G_i$;
\item $G[\bigcup A^G_i]$ is a complete multipartite graph, with partition $A^G_1,\ldots,A^G_p$;
\item For every $i \in \{1,\ldots,p\}$, every vertex of~$N^G_i$ has a neighbour in~$A^G_i$, but is anti-complete to~$A^G_j$ for every $j \in \{1,\ldots,p\} \setminus \{i\}$; and
\item $B^G$ is anti-complete to~$\bigcup A^G_i$.
\end{enumerate}
Furthermore, given~$K^G$, this partition is unique (up to permuting the indices on the~$A^G_i$s and corresponding~$N^G_i$s) and can be found in polynomial time.
\end{lemma}

\begin{proof}
Let~$G$ be a $\overline{2P_1+P_3}$-free graph containing an induced~$K_5$ with vertex set~$K^G$.
If a vertex $v \in V(G) \setminus K^G$ has two neighbours $x,x' \in K^G$ and two non-neighbours $y,y' \in K^G$, then $G[x,x',y,v,y']$ is a~$\overline{2P_1+P_3}$, a contradiction.
Therefore every vertex in $V(G) \setminus K^G$ has either at most one non-neighbour in~$K^G$ or at most one neighbour in~$K^G$.
Let~$L^G$ denote the set of vertices that are either in~$K^G$ or have at most one non-neighbour in~$K^G$ and note that~$L^G$ is uniquely defined by the choice of~$K^G$.

We claim that~$G[L^G]$ is a complete multipartite graph.
Suppose, for contradiction, that~$G[L^G]$ is not complete multipartite.
Then~$G[L^G]$ contains an induced $P_1+\nobreak P_2=\overline{P_3}$, say on vertices $v,v',v''$ (note that
some of these vertices may be in~$K^G$).
Now each of $v,v',v''$ has at most one non-neighbour in~$K^G$ and if a vertex $w \in \{v,v',v''\}$ is in~$K^G$, then it is adjacent to every vertex in~$K^G \setminus \{w\}$.
Therefore, since $|K^G|=5$, there must be vertices $u,u'\in K^G\setminus \{v,v',v''\}$ that are complete to $\{v,v',v''\}$.
Now $G[u,u',v',v,v'']$ is a~$\overline{2P_1+P_3}$.
This contradiction completes the proof that~$G[L^G]$ is complete multipartite.

We let $A^G_1,\ldots,A^G_p$ be the partition classes of the complete multipartite graph~$G[L^G]$. 
Note that $p \geq 5$, since each~$A^G_i$ contains at most one vertex of~$K^G$.
We claim that each vertex not in~$L^G$ has neighbours in at most one set~$A^G_i$.
Suppose, for contradiction, that there is a vertex $v \in V(G) \setminus L^G$ with neighbours in two distinct sets~$A^G_i$, say~$v$ is adjacent to $u \in A^G_1$ and $u' \in A^G_2$.
Since $v \notin L^G$, the vertex~$v$ has at most one neighbour in~$K^G$.
Since $|K^G|=5$, there must be two vertices $y,y' \in K^G \setminus (A^G_1 \cup A^G_2)$ that are non-adjacent to~$v$.
Now $G[u,u',y,v,y']$ is a~$\overline{2P_1+P_3}$, a contradiction.
Therefore every vertex not in~$L^G$ has neighbours in at most one set~$A^G_i$.
Let~$N^G_i$ be the set of vertices in $V(G) \setminus L^G$ that have neighbours in~$A^G_i$ and let~$B^G$ be the set of vertices in $V(G) \setminus L^G$ that are anti-complete to~$L^G$.
Finally, note that the partition of~$V(G)$ into sets $A^G_1,\ldots,A^G_p,N^G_1,\ldots,N^G_p,B^G$ can be found in polynomial time and is unique (up to permuting the indices on the~$A^G_i$s and corresponding~$N^G_i$s).\qed
\end{proof}

For the $(\overline{2P_1+P_3},P_5)$-free case, we will use the following observation.

\begin{sloppypar}
\begin{observation}\label{obs:add-false-twin-is-safe}
If~$G$ is a graph containing a vertex~$x$ and~$G^x$ is the graph obtained from~$G$ by adding a new vertex~$x'$ with the same neighbourhood as~$x$, then~$G^x$ is $(\overline{2P_1+P_3},P_5)$-free if and only if~$G$ is $(\overline{2P_1+P_3},P_5)$-free.
\end{observation}
\end{sloppypar}
\begin{proof}
Since~$G$ is an induced subgraph of~$G^x$, if~$G^x$ is $(\overline{2P_1+P_3},P_5)$-free then~$G$ is $(\overline{2P_1+P_3},P_5)$-free.
Suppose, for contradiction, that~$G^x$ contains a set of vertices~$X$ that induce a~$\overline{2P_1+P_3}$ or a~$P_5$, but that~$G$ is $(\overline{2P_1+P_3},P_5)$-free.
Since neither~$\overline{2P_1+P_3}$ nor~$P_5$ has two vertices with the same neighbourhood, it follows that either $x \notin X$ or $x' \notin X$.
By symmetry, we may assume that $x' \notin X$, in which case $X \subseteq V(G)$, so~$G$ contains an induced~$\overline{2P_1+P_3}$ or~$P_5$, a contradiction.\qed
\end{proof}

\begin{theorem}\label{thm:coP5-2P1+P3-easy}
{\sc Graph Isomorphism} is polynomial-time solvable on $(\overline{2P_1+P_3},P_5)$-free graphs.
\end{theorem}
\begin{proof}
Since {\sc Graph Isomorphism} can be solved component-wise, we need only consider connected graphs.
Therefore, since {\sc Graph Isomorphism} is polynomial-time solvable on $(K_5,P_5)$-free graphs by Lemma~\ref{lem:K_t-P_5-free-gi-poly}, and we can test whether a graph is $K_5$-free in polynomial time, it only remains to consider the class of connected $(\overline{2P_1+P_3},P_5)$-free graphs~$G$ that contain an induced~$K_5$.
Let~$K^G$ be the vertices of such a~$K_5$ in~$G$.
Let $A^G_1,\ldots,A^G_p,N^G_1,\ldots,N^G_p,B^G$ be defined as in Lemma~\ref{lem:K5-extension} and let $L^G = \bigcup A^G_i$.
We start by proving the following claim.

\clm{\label{clm2:3Ni}If at least three~$N^G_i$'s are non-empty, then $G$ has bounded clique-width.}
Suppose, for contradiction, that we can find vertices $x,y,z$ with $x \in N^G_i, y \in N^G_j, z \in N^G_k$ with $i,j,k$ pairwise distinct such that~$z$ is adjacent to~$y$, but not to~$x$.
Let $x' \in A^G_i$ be a neighbour of~$x$, let $y' \in A^G_j$ be a neighbour of~$y$ 
and let $x'' \in K^G \setminus (A^G_i \cup A^G_j \cup A^G_k)$
(which exists since the sets~$A^G_i$, $A^G_j$ and~$A^G_k$ each contain at most one vertex of~$K^G$, while $|K^G|=5$).
Then $G[x'',x',x,y,z]$ or $G[x,x',y',y,z]$ is a~$P_5$ if~$x$ is adjacent or non-adjacent to~$y$, respectively.
It follows that the~$N^G_i$'s are either pairwise anti-complete or pairwise complete.
We consider these two cases separately.

\thmcase{\label{case2:2Ni-anti}At least three~$N^G_i$'s are non-empty and the~$N^G_i$'s are pairwise anti-complete.}
Suppose, for contradiction, that there is a vertex $x \in B^G$.
Since~$G$ is connected, $x$ must have a neighbour $y \in N^G_i$ for some~$i$.
Choose a vertex $z \in N^G_j$ for some $j \neq i$ and note that~$z$ is non-adjacent to~$y$.
Let $y' \in A^G_i$ and $z' \in A^G_j$ be neighbours of~$y$ and~$z$, respectively, and let $w \in K^G \setminus (A^G_i \cup A^G_j)$.
Then $G[w,y',y,x,z]$ or $G[x,y,y',z',z]$ is a~$P_5$ if~$x$ is adjacent or non-adjacent to~$z$, respectively.
This contradiction implies that $B^G = \emptyset$.

Suppose, for contradiction, that there are two adjacent vertices $y,y' \in N^G_i$ that have different neighbourhoods in~$A^G_i$, say~$y$ is adjacent to $z \in A^G_i$, but~$y'$ is not.
Let $x \in N^G_j$ for some $j \neq i$ and let $x' \in A^G_j$ be a neighbour of~$x$; note that~$x$ is non-adjacent to~$y$ and~$y'$.
Then $G[x,x',z,y,y']$ is a~$P_5$.
This contradiction implies that if two vertices in some set~$N^G_i$ are in the same component of~$G[N^G_i]$, then they must have the same neighbourhood in~$A^G_i$.

Suppose, for contradiction, that for some~$i$ there are vertices $x,y \in N^G_i$ with incomparable neighbourhoods in~$A^G_i$.
Note that in this case~$x$ and~$y$ must be in different components of~$G[N^G_i]$, so they must be non-adjacent to each other.
Let $x' \in A^G_i$ be a neighbour of~$x$ that is non-adjacent to~$y$, let $y' \in A^G_i$ be a neighbour of~$y$ that is non-adjacent to~$x$ and let $z \in K^G \setminus A^G_i$.
Then $G[x,x',z,y',y]$ is a~$P_5$.
This contradiction implies that the components of~$G[N^G_i]$ can be ordered by containment of their neighbourhoods in~$A^G_i$.

We will now show that $G_i:=G[A^G_i \cup N^G_i]$ has bounded clique-width.
We may order the vertices of~$A^G_i$, say $a_1,\ldots,a_r$, in decreasing order of neighbourhoods in~$N^G_i$ (breaking ties arbitrarily); note that every vertex of~$N^G_i$ is adjacent to~$a_1$.
We partition~$N^G_i$ into sets $X_1,\ldots,X_r$ such that the vertices of~$X_j$ are adjacent to~$a_k$ if and only if $k \leq j$.
Since~$a_1$ dominates~$N^G_i$, and~$G[N^G_i\cup\{a_1\}]$ is $\overline{2P_1+P_3}$-free, it follows that $G[N^G_i]$ is $\paw$-free (recall that the $\paw$ is $\overline{P_1+P_3}$).
By Corollary~\ref{cor:pawP5-free-bdd-cw}, it follows that~$G[N^G_i]$ has bounded clique-width.
Therefore, for some constant~$c$, we can construct each of $G[X_1],\ldots,G[X_r]$ using only labels from~$\{1,\ldots,c\}$.
We will now construct~$G_i$ using two new labels~$1'$ and~$2'$ in addition to the labels from~$\{1,\ldots,c\}$.
For $j \in \{1,\ldots,r\}$, suppose we have constructed $G_i[X_1,\ldots,X_{j-1} \cup \{a_1,\ldots,a_{j-1}\}]$ such that the vertices in $X_1,\ldots,X_{j-1}$ have label~$1'$ and the vertices in $\{a_1,\ldots,a_{j-1}\}$ have label~$2'$ (if $j=1$, this means we have constructed the empty graph).
We then construct~$G[X_j]$ using labels from~$\{1,\ldots,c\}$ and construct~$a_j$ with label~$2'$ and take the disjoint union of these and the graph constructed so far.
We join vertices with labels in~$\{1,\ldots,c\}$ to the vertices with label~$2'$ and then relabel the vertices with label~$\{1,\ldots,c\}$ to have label~$1'$.
We have now constructed $G_i[X_1,\ldots,X_j \cup \{a_1,\ldots,a_j\}]$ such that the vertices in $X_1,\ldots,X_j$ have label~$1'$ and the vertices in $\{a_1,\ldots,a_j\}$ have label~$2'$.
By induction, we can therefore construct~$G_i$ with~$c+2$ labels.
It follows that~$G_i$ has bounded clique-width.

Now, for every~$i$, let~$G_i^*$ be the graph obtained from~$G_i$ by complementing~$A^G_i$ and note that~$G_i^*$ has bounded clique-width by Fact~\ref{fact:comp}.
Let~$G^*$ be the disjoint union of the~$G_i^*$ graphs.
Note that~$G$ is the graph obtained from~$G^*$ by complementing~$L^G$.
By Fact~\ref{fact:comp}, it follows that~$G$ has bounded clique-width.
This completes Case~\ref{case2:2Ni-anti}.

\thmcase{\label{case2:3Ni-comp}At least three~$N^G_i$'s are non-empty and the~$N^G_i$'s are pairwise complete.}
We first claim that~$B^G$ is complete to~$\bigcup N^G_i$.
Suppose, for contradiction, that there is a vertex in $x \in B^G$ that has both a neighbour~$y$ and a non-neighbour~$z$ in $\bigcup N^G_i$.
Since there is more than one non-empty set~$N^G_i$, we may assume that $y \in N^G_i$ and $z \in N^G_j$ for some $i \neq j$; note that this means~$y$ is adjacent to~$z$.
Let $z' \in A^G_j$ be a neighbour of~$z$ and let $z'' \in K^G \setminus (A^G_i \cup A^G_j)$.
Then~$G[x,y,z,z',z'']$ is a~$P_5$, a contradiction, and so every vertex of~$B^G$ is either complete or anti-complete to~$\bigcup N^G_i$.
Since~$G$ is connected, if not every vertex of~$B^G$ is complete to~$\bigcup N^G_i$, then there must be adjacent vertices $x,x' \in B^G$ that are complete and anti-complete to~$\bigcup N^G_i$, respectively.
Let $y \in N^G_i$ for some~$i$, let $y' \in A^G_i$ be a neighbour of~$y$ and let $z \in K^G \setminus A^G_i$.
Then $G[x',x,y,y',z]$ is a~$P_5$.
This contradiction implies that~$B^G$ is indeed complete to~$\bigcup N^G_i$.

Now suppose, for contradiction, that for some~$i$ there is a vertex $z \in N^G_i$ that has a non-neighbour $x \in A^G_i$.
Let $x' \in A^G_i$ be a neighbour of~$z$, let $z' \in N^G_j$ for some $j \neq i$ and let $y \in K^G \setminus (A^G_i \cup A^G_j)$.
Then $G[x,y,x',z,z']$ is a~$P_5$, a contradiction.
It follows that for every~$i$, $N^G_i$ is complete to~$A^G_i$.

Now~$B^G$ is dominated by a vertex of~$N^G_i$ for some~$i$.
Moreover, for every~$i$ the set~$N^G_i$ is dominated by a vertex in~$A^G_i$.
Since~$G$ is a $\overline{2P_1+P_3}$-free graph, it follows that~$G[B^G]$ and, for every~$i$,~$G[N^G_i]$ are $\paw$-free graphs and thus have bounded clique-width by Corollary~\ref{cor:pawP5-free-bdd-cw}.
Since~$G[A^G_i]$ is an edgeless graph for every~$i$, it has clique-width~$1$.
The graph $G_i:=G[N^G_i \cup A^G_i]$ can be obtained from~$G[N^G_i]$ and~$G[A^G_i]$  by taking their disjoint union and applying a bipartite complementation between~$N^G_i$ and~$A^G_i$.
By Fact~\ref{fact:bip}, it follows that~$G_i$ has bounded clique-width.

Let~$G^*_i$ be the graph obtained from~$G_i$ by complementing~$A^G_i$ and~$N^G_i$.
Then~$G^*_i$ has bounded clique-width by Fact~\ref{fact:comp}.
We take the disjoint union of the~$G^*_i$'s and~$G[B^G]$.
Now, if we complement~$L^G$ and~$\bigcup N^G_i$ and apply a bipartite complementation between~$B^G$ and~$\bigcup N^G_i$ we obtain the graph~$G$.
By Facts~\ref{fact:comp} and~\ref{fact:bip}, it follows that~$G$ has bounded clique-width.
This completes Case~\ref{case2:3Ni-comp} and therefore completes the proof of Claim~\ref{clm2:3Ni}.\dia

\medskip
\noindent
We now describe an algorithm to prove Theorem~\ref{thm:coP5-2P1+P3-easy}.
Suppose~$G$ and~$H$ are $(\overline{2P_1+P_3},P_5)$-free graphs.
We can enumerate all sets~$K^G$ that induce a $K_5$ in~$G$ in polynomial time.
By Lemma~\ref{lem:K5-extension}, we can therefore test in polynomial time whether there is a~$K^G$ such that at least three~$N^G_i$ sets are non-empty; if so, then~$G$ has bounded clique-width by Claim~\ref{clm2:3Ni} and we apply Theorem~\ref{thm:gi-poly-bdd-cw}.

We may now assume that for every~$K^G$ at most two sets~$N^G_i$ are non-empty.
We may also assume that the same is true for every~$K^H$ in~$H$ (otherwise we immediately output that~$G$ and~$H$ are not isomorphic).
We will now explain how to transform~$G$ into a graph~$G'$ that is $K_5$-free.

First note that if $x \in A^G_i$ for some~$i$ such that $N^G_i=\emptyset$, then $L^G=A^G_i \cup N(x)$.
Since~$A^G_i$ is the set of vertices in~$G$ with the same neighbourhood as~$x$, every set~$L^G$ can be written as $N(x)\cup\{y \;|\; N(y)=N(x)\}$ for some vertex~$x$ of~$G$.
Moreover, for every choice of~$L^G$, there are at least three sets~$A^G_i$ such that~$N_i^G=\emptyset$.
Now $L^G=N(x)\cup\{y \;|\; N(y)=N(x)\}$ holds for every vertex~$x$ in such a set~$A_i^G$, so every~$L^G$ can be obtained in this way from at least three possible vertices~$x$.
We conclude that there are at most~$\frac{n}{3}$ possible sets~$L^G$.

Given a set~$L^G$, recall that the sets~$A_i^G$ are uniquely determined (up to reordering).
Let~$L'^G$ denote the set $\bigcup_{i \; | \; N^G_i=\emptyset}A^G_i$; we say that the multiset $\{|A^G_i| \; | \; N^G_i=\emptyset\}$ is the \emph{type} of~$L^G$.
We consider all possibilities for~$L^G$ in~$G$ and number the different types that occur $1,\ldots,t$; note that the possible sets~$L'^G$ are pairwise vertex-disjoint.
Suppose that for $j \in \{1,\ldots,t\}$ we replace the vertices of~$L'^G$ in each set~$L^G$ of type~$j$ by a copy of~$K_{n+j,n+j}$ that is complete to~$L^G \setminus L'^G$, where~$n$ denotes the number of vertices in the original graph~$G$.
Note that since~$L'^G$ is a complete multipartite graph with at least three parts, this would change the graph in the same way as deleting all but two parts of this multipartite graph and then expanding the remaining two parts by adding false twins of vertices already in the graph.
By Observation~\ref{obs:add-false-twin-is-safe}, the resulting graph~$G'$ is still $(\overline{2P_1+P_3},P_5)$-free.
Furthermore, applying this operation removes every~$K_5$ from the graph, so~$G'$ is a $(K_5,P_5)$-free graph.
We can apply the same transformation to~$H$ to obtain a $(K_5,P_5)$-free graph~$H'$.
For~$G$ and~$H$ we can enumerate all possible sets~$L'^G$ and~$L'^H$ (using Lemma~\ref{lem:K5-extension}), and, as observed above, there are at most~$\frac{n}{3}$ such sets in each graph.
For each type of an~$L'^G$ in~$G$, $H$ must have the same number of sets~$L'^H$ with this type as~$G$ does (and vice verse), otherwise we output that~$G$ and~$H$ are not isomorphic.
We therefore number the types of~$L'^G$ in~$G$ and the types~$L'^H$ in~$H$ in the same way.
Since~$G'$ and~$H'$ are $(K_5,P_5)$-free graphs, by Lemma~\ref{lem:K_t-P_5-free-gi-poly}, we can test whether they are isomorphic in polynomial time.

It therefore suffices to show that~$G'$ and~$H'$ are isomorphic if and only if~$G$ and~$H$ are isomorphic.
By construction, if~$G$ and~$H$ are isomorphic, then~$G'$ and~$H'$ are isomorphic.
Now suppose that there is an isomorphism~$f$ from~$G'$ to~$H'$.
For a vertex $x \in V(G')$, let $V^{G'}_x=\{v \in V(G') \; | \; N(v)=N(x)\}$.
Note that $|V^{H'}_{f(x)}|=|f(V^{G'}_x)|=|V^{G'}_x|$.
Now $x \in V(G') \setminus V(G)$ if and only if $|V^{G'}_x|>n$.
By construction, two sets of the form~$V^{G'}_x$ with $|V^{G'}_x|>n$ are either complete or anti-complete to each other and each such set is complete to exactly one other such set.
Therefore, for each $j \in \{1,\ldots,t\}$, the isomorphism~$f$ maps the copies of~$K_{n+j,n+j}$ from the construction of~$G'$ to copies of~$K_{n+j,n+j}$ from the construction of~$H'$ and~$f$ maps $V(G) \cap V(G')$ to $V(H) \cap V(H')$.
We may therefore replace each copy of~$K_{n+j,n+j}$ in~$G'$ and~$H'$ by an~$L'^G$ and~$L'^H$ of the corresponding type.
Since it is trivial to find an isomorphism from a set~$L'^G$ to a set~$L'^H$ of the same type, we can construct an isomorphism from~$G$ to~$H$.\qed
\end{proof}

We are now ready to prove Theorem~\ref{thm:co2P1+P3-2P1+P3-easy}.
Note that as $P_2+\nobreak P_3$ contains two vertices with the same neighbourhood, we do not have an analogue of Observation~\ref{obs:add-false-twin-is-safe} for the $(\overline{2P_1+P_3},\allowbreak P_2+\nobreak P_3)$-free case.
Because of this, the proof of Theorem~\ref{thm:co2P1+P3-2P1+P3-easy} is slightly more involved than that of Theorem~\ref{thm:coP5-2P1+P3-easy}.

\begin{sloppypar}
\begin{theorem}\label{thm:co2P1+P3-2P1+P3-easy}
{\sc Graph Isomorphism} is polynomial-time solvable on $(\overline{2P_1+P_3},\allowbreak P_2+\nobreak P_3)$-free graphs.
\end{theorem}
\end{sloppypar}
\begin{proof}
Since {\sc Graph Isomorphism} can be solved component-wise, we need only consider connected graphs.
Therefore, as {\sc Graph Isomorphism} is polynomial-time solvable on $(K_5,P_2+\nobreak P_3)$-free graphs by Lemma~\ref{lem:K_t-K_1t-free-gi-poly}, and we can test whether a graph is $K_5$-free in polynomial time, it only remains to consider the class of connected $(\overline{2P_1+P_3},P_2+\nobreak P_3)$-free graphs~$G$ that contain an induced~$K_5$.
Let~$K^G$ be the vertices of an induced~$K_5$ in~$G$ (note that such a set~$K^G$ can be found in polynomial time, but it is not necessarily unique).
Let $A^G_1,\ldots,A^G_p,N^G_1,\ldots,N^G_p,B^G$ be defined as in Lemma~\ref{lem:K5-extension} and let $L^G = \bigcup A^G_i$ and $D^G = V(G) \setminus L^G$.

Now suppose that~$G$ and~$H$ are connected $(\overline{2P_1+P_3},P_2+\nobreak P_3)$-free graphs that each contain an induced~$K_5$.
If~$G$ and~$H$ have bounded clique-width (which happens in Case~\ref{case1:D-2P2-free} below), then by Theorem~\ref{thm:gi-poly-bdd-cw} we are done.
Otherwise, note that if~$K^G$ and~$K^H$ are vertex sets that induce a~$K_5$ in~$G$ and~$H$, respectively, then Lemma~\ref{lem:K5-extension} implies that $L^G, D^G, L^H$ and~$D^H$ are uniquely defined.
Therefore, we fix one choice of~$K^G$ and, for each choice of~$K^H$, test whether there is an isomorphism $f:G \to H$ such that $f(L^G)=L^H$ (we use this approach in Cases~\ref{case1:B-K4-free} and~\ref{case1:B-K4-non-free} below).
Clearly, we may assume that the vertex partitions given by Lemma~\ref{lem:K5-extension} for~$G$ and~$H$ have the same value of~$p$ and that $|A^G_i|=|A^H_i|$ and $|N^G_i|=|N^H_i|$ for all $i\in \{1,\ldots,p\}$ and $|B^G|=|B^H|$.
Furthermore, for any claims we prove about~$G$ and its vertex sets, we may assume that the same claims hold for~$H$ (otherwise such an isomorphism~$f$ does not exist).
We start by proving the following four claims.

\clm{\label{clm1:BUNi-P3-free}$G[D^G]$ is $P_3$-free.}
Indeed, suppose, for contradiction, that~$G[D^G]$ contains an induced~$P_3$, say on vertices $u,u',u''$.
Since $|K^G|=5$ and each vertex in~$D^G$ has at most one neighbour in~$K^G$, there must be vertices $v,v' \in K^G$ that are anti-complete to $\{u,u',u''\}$.
Then $G[v,v',u,u',u'']$ is a $P_2+\nobreak P_3$, a contradiction.\dia

\clm{\label{clm1:no-special-P1P2}If $v \in N^G_j$ for some $j \in \{1,\ldots,p\}$ and there are two adjacent vertices $u,u' \in D^G\setminus N^G_j$, then~$v$ is complete to~$\{u,u'\}$.}
Since~$G[D^G]$ is $P_3$-free by Claim~\ref{clm1:BUNi-P3-free}, the vertex~$v$ must be either complete or anti-complete to~$\{u,u'\}$.
Suppose, for contradiction, that~$v$ is anti-complete to~$\{u,u'\}$.
Since $v \in N^G_j$, $v$ has a neighbour $v' \in A^G_j$.
Since $|K^G\setminus A^G_j|\geq 4$ and each vertex in~$D^G$ has at most one neighbour in~$K^G$, there is a vertex $v'' \in K^G \setminus A^G_j$ that is non-adjacent to both~$u$ and~$u'$.
Since $v'' \notin A^G_j$, $v''$ is also non-adjacent to~$v$, but is adjacent to~$v'$.
Now $G[u,u',v,v',v'']$ is a~$P_2+\nobreak P_3$, a contradiction.\dia

\clm{\label{clm:P1+C3-structure}If~$G[D^G]$ has at least two components and one of these components~$C$ has at least three vertices, then there is an $i \in \{1,\ldots,p\}$ such that $D^G \setminus C \subset N^G_i \cup B^G$ and all but at most one vertex of~$C$ belongs to~$N^G_i$.}
By Claim~\ref{clm1:BUNi-P3-free}, $G[D^G]$ is a disjoint union of cliques.
Since~$G$ is connected, $D^G \setminus C$ cannot be a subset of~$B^G$.
Hence, for some $i \in \{1,\ldots,p\}$, there must be a vertex $x \in N^G_i \setminus C$.
Therefore, by Claim~\ref{clm1:no-special-P1P2}, at most one vertex of~$C$ can lie outside of~$N^G_i$.
Since $|C|\geq 3$, it follows that $C \cap N^G_i$ contains at least two vertices.
Since the vertices in~$C$ are pairwise adjacent, by Claim~\ref{clm1:no-special-P1P2} it follows that $D^G \setminus C \subset N^G_i \cup B^G$.\dia

\clm{\label{clm1:D-2P2-non-free-vertex-nonadj-to-P2-is-anticomplete}Let $i \in \{1,\ldots,p\}$.
If~$G[D^G]$ contains at least two non-trivial components and there is a vertex~$v$ in~$A^G_i$ with two non-neighbours in the same component of~$G[D^G]$, then~$v$ is anti-complete to~$D^G$.
Furthermore, there is at most one vertex in~$A^G_i$ with this property.}
Suppose $v \in A^G_i$ has two non-neighbours $x,x'$ in some component~$C$ of~$G[D^G]$.
By Claim~\ref{clm1:BUNi-P3-free}, $G[D^G]$ is a disjoint union of cliques, so~$x$ must be adjacent to~$x'$.
We claim that~$v$ is anti-complete to~$D^G \setminus C$.
Suppose, for contradiction, that~$v$ has a neighbour $y \in D^G \setminus C$.
Since every vertex of~$D^G$ has at most one neighbour in~$K^G$, there must be a vertex $z \in K^G \setminus A_i^G$ that is non-adjacent to~$x,x'$ and~$y$ and so $G[x,x',y,v,z]$ is a~$P_2+\nobreak P_3$.
This contradiction implies that~$v$ is indeed anti-complete to $D^G \setminus C$.
Now $G[D^G \setminus C]$ contains another non-trivial component~$C'$ and we have shown that~$v$ is anti-complete to~$C'$.
Repeating the same argument with~$C'$ taking the place of~$C$, we find that~$v$ is anti-complete to~$D^G \setminus C'$, and therefore~$v$ is anti-complete to~$D^G$.
Finally, suppose, for contradiction, that there are two vertices $v,v' \in A^G_i$ that are both anti-complete to~$D^G$.
Let $x,x'$ be adjacent vertices in~$D^G$ and let $z \in K^G\setminus A^G_i$ be a vertex non-adjacent to~$x$ and~$x'$.
Then $G[x,x',v,z,v']$ is a~$P_2+\nobreak P_3$, a contradiction.\dia

\medskip
\noindent
We now start a case distinction and first consider the following case.

\thmcase{\label{case1:D-2P2-free}$G[D^G]$ contains at most one non-trivial component.}
In this case we will show that~$G$ has bounded clique-width, and so we will be done by Theorem~\ref{thm:gi-poly-bdd-cw}.
By Claim~\ref{clm1:BUNi-P3-free}, every component of~$G[D^G]$ is a clique.
Since~$G[D^G]$ contains at most one non-trivial component, we may partition~$D^G$ into a clique~$C$ and an independent set~$I$ (note that~$C$ or~$I$ may be empty).
If~$|C| \geq 3$ and $|I| \geq 1$, then by Claim~\ref{clm:P1+C3-structure} there is an $i \in \{1,\ldots,p\}$ such that at most one vertex of $C \cup I$ is outside~$N^G_i$; if such a vertex exists, then by Fact~\ref{fact:del-vert} we may delete it.
Now if $|C| \leq 3$, then by Fact~\ref{fact:del-vert} we may delete the vertices of~$C$.
Thus we may assume that either $C=\emptyset$ or $|C|\geq 4$ and furthermore, if $|C|\geq 4$ and $|I|\geq 1$, then $C \cup I \subseteq N^G_i$ for some $i \in \{1,\ldots,p\}$.
Note that $I \cap B^G = \emptyset$ since~$G$ is connected, so $B^G \subset C$.
Therefore~$G[B^G]$ is a complete graph, so it has clique-width at most~$2$.
Applying a bipartite complementation between~$B^G$ and~$C \setminus B^G$ removes all edges between~$B^G$ and~$V(G) \setminus B^G$.
By Fact~\ref{fact:bip}, we may therefore assume that $B^G=\emptyset$.

Let~$M$ be the set of vertices in~$L^G$ that have neighbours in~$I$.
We claim that~$M$ is complete to all but at most one vertex of~$C$.
We may assume that $|C|\geq 4$ and $|I|\geq 1$, otherwise the claim follows trivially.
Therefore, as noted above, $C \cup I \subseteq N^G_i$ for some $i \in \{1,\ldots,p\}$.
Suppose $u \in M$ has a neighbour $u' \in I$ and note that this implies $u \in A^G_i$, $u' \in N^G_i$.
Suppose, for contradiction, that~$u$ has two non-neighbours $v,v' \in C$ and let $w \in K^G \setminus A^G_i$.
Then $G[v,v',u',u,w]$ is a~$P_2+\nobreak P_3$, a contradiction.
Therefore if $u \in M$, then~$u$ has at most one non-neighbour in~$C$.
Now suppose that there are two vertices $u,u' \in M$. 
It follows that $u,u'\in A^G_i$, so these vertices must be non-adjacent.
Furthermore, each of these vertices has at most one non-neighbour in~$C$.
If~$u$ and~$u'$ have different neighbourhoods in~$C$, then without loss of generality we may assume that there are vertices $x,y,y' \in C$ such that~$u$ is adjacent to~$x,y$ and~$y'$ and~$u'$ is adjacent to~$y$ and~$y'$, but not to~$x$.
Now $G[y,y',u,u',x]$ is a~$\overline{2P_1+P_3}$, a contradiction.
Therefore every vertex in~$M$ has the same neighbourhood in~$C$, which consists of all but at most one vertex of~$C$ and the claim holds.
If the vertices of~$M$ are not complete to~$C$, then we delete one vertex of~$C$ (we may do so by Fact~\ref{fact:del-vert}), after which~$M$ will be complete to~$C$.
We may therefore assume that~$M$ is complete to~$C$.

Now note that for all $i \in \{1,\ldots,p\}$, the graph $G_i=G[(A^G_i\setminus M) \cup (N^G_i \cap C)]$ is a $\overline{2P_1+P_3}$-free split graph, so it has bounded clique-width by Lemma~\ref{lem:2P_1P3-split-bdd-cw}.
Furthermore $G_i'=G[(A^G_i\cap M) \cup (N^G_i \cap I)]$ is a $(P_2+\nobreak P_3)$-free bipartite graph, so it has bounded clique-width by Lemma~\ref{lem:K3-P6-bdd-cw}.
Let~$G_i''$ be the graph obtained from the disjoint union $G_i+\nobreak G_i'$ by complementing~$A^G_i$ and $(N^G_i \cap C)$.
By Fact~\ref{fact:comp}, $G_i''$ also has bounded clique-width.
Therefore the disjoint union~$G^*$ of all the~$G_i''$s has bounded clique-width.
Now~$G$ can be constructed from~$G^*$ by complementing~$L^G$, complementing~$C$ and applying a bipartite complementation between~$C$ and~$M$.
Hence, by Facts~\ref{fact:comp} and~\ref{fact:bip}, $G$ has bounded clique-width.
This completes Case~\ref{case1:D-2P2-free}.

\medskip
\noindent
We may now assume that Case~\ref{case1:D-2P2-free} does not apply, that is, $G[D^G]$ has at least two non-trivial components.
This leads us to our second and third cases.

\thmcase{\label{case1:B-K4-free}$G[D^G]$ contains at least two non-trivial components, but is $K_4$-free.}
Recall that~$G[D^G]$ is $P_3$-free by Claim~\ref{clm1:BUNi-P3-free}, so every component of~$G[D^G]$ is a clique.
Let~$C$ be a non-trivial component of~$G[D^G]$ and let $x,y \in C$.
Then~$x$ is adjacent to~$y$ and $x,y \in N^G_i \cup N^G_j \cup B^G$ for some (not necessarily distinct) $i,j \in \{1,\ldots,p\}$.
By Claim~\ref{clm1:no-special-P1P2}, every vertex~$z$ in a component of~$G[D^G]$ other than~$C$ must also be in $N^G_i \cup N^G_j \cup B^G$.
Since~$G[D^G]$ contains at least two non-trivial components, repeating this argument with another non-trivial component implies that every vertex of~$D^G$ lies in $N^G_i \cup N^G_j \cup B^G$.
Without loss of generality, we may therefore assume that $N^G_k = \emptyset$ for $k \geq 3$.

Since~$G[D^G]$ is $K_4$-free, for each $i \in \{1,\ldots,p\}$ the graph $G[D^G \cup A^G_i]$ is $K_5$-free.
This means that every~$K_5$ in~$G$ is entirely contained in~$L^G$.
By Claim~\ref{clm1:D-2P2-non-free-vertex-nonadj-to-P2-is-anticomplete}, for $i \geq 3$, $|A^G_i|=1$ and so
 $L^G \setminus (A^G_1 \cup A^G_2)$ must be a clique.
The vertices of $L^G \setminus (A^G_1 \cup A^G_2)$ have no neighbours outside~$L^G$
and are adjacent to every other vertex of~$L^G$, so these vertices are in some sense interchangeable.
Indeed, $N[v]=L^G$ for every $v \in L^G \setminus (A^G_1 \cup A^G_2)$, and so every bijection that permutes the vertices of $L^G \setminus (A^G_1 \cup A^G_2)$ and leaves the other vertices of~$G$ unchanged is an isomorphism from~$G$ to itself.
Let~$G'$ be the graph obtained from~$G$ by deleting all vertices in~$A^G_i$ for $i \geq 6$ (if any such vertices are present).
Now~$G'$ is $K_6$-free, so it is a $(K_6,P_2+\nobreak P_3)$-free graph.
Therefore we can test isomorphism of such graphs~$G'$ in polynomial time by Lemma~\ref{lem:K_t-K_1t-free-gi-poly}.
If there is an isomorphism between two such graphs~$G'$ and~$H'$, then, because the vertices of $L^G \setminus (A^G_1 \cup A^G_2)$ are interchangeable, we can extend it to a full isomorphism of~$G$ and~$H$ by mapping the remaining vertices of $L^G \setminus (A^G_1 \cup A^G_2)$ to $L^H \setminus (A^H_1 \cup A^H_2)$ arbitrarily.
This completes Case~\ref{case1:B-K4-free}.

\thmcase{\label{case1:B-K4-non-free}$G[D^G]$ contains at least two non-trivial components and contains an induced~$K_4$.}
Recall that~$G[D^G]$ is $P_3$-free by Claim~\ref{clm1:BUNi-P3-free}, so every component of~$G[D^G]$ is a clique.
We claim that $D^G \subseteq N^G_i \cup B^G$ for some $i \in \{1,\ldots,p\}$.
Let~$C$ be a component of~$G[D^G]$ that contains at least four vertices, and let~$C'$ be a component of~$G[D^G]$ other than~$C$, and note that such components exist by assumption.
By Claim~\ref{clm:P1+C3-structure}, there is an $i \in \{1,\ldots,p\}$ such that $D^G \setminus C \subset N^G_i \cup B^G$ and all but at most one vertex of~$C$ belongs to~$N^G_i$.
In particular, this implies that $C' \subset N^G_i \cup B^G$.
By Claim~\ref{clm1:no-special-P1P2}, it follows that~$C$ cannot have a vertex in~$N^G_j$ for some $j \in \{1,\ldots,p\} \setminus \{i\}$, and so $C \subset N^G_i \cup B^G$.
Without loss of generality, we may therefore assume that $N^G_j=\emptyset$ for $j \in \{2,\ldots,p\}$ and so $D^G=N^G_1 \cup B^G$.
Now if $j \in \{2,\ldots,p\}$, then the vertices of~$A^G_j$ are anti-complete to~$D^G$, so Claim~\ref{clm1:D-2P2-non-free-vertex-nonadj-to-P2-is-anticomplete} implies that $|A^G_j|=1$.
This implies that $L^G\setminus A^G_1$ is a clique.

By Claim~\ref{clm1:D-2P2-non-free-vertex-nonadj-to-P2-is-anticomplete} there is at most one vertex $x^G \in A^G_1$ that has two non-neighbours in the same non-trivial component~$C$ of~$G[D^G]$ and if such a vertex exists, then it must be anti-complete to~$D^G$.
Let $A^{*G}_1=A^G_1 \setminus \{x^G\}$ if such a vertex~$x^G$ exists and $A^{*G}_1=A^G_1$ otherwise.
Then every vertex in~$A^{*G}_1$ has at most one non-neighbour in each component of~$G[D^G]$.
Note that~$A^{*G}$ is non-empty, since~$D^G$ is non-empty and~$G$ is connected.

Suppose~$C$ is a component of~$G[D^G]$ on at least four vertices.
Now suppose, for contradiction, that there are two vertices $y,y' \in A^{*G}_1$ with different neighbourhoods in~$C$.
Then without loss of generality there is a vertex $x \in C$ that is adjacent to~$y$, but not to~$y'$.
Since $|C| \geq 4$ and every vertex in~$A^{*G}_1$ has at most one non-neighbour in~$C$, there must be two vertices $z,z' \in C$ that are adjacent to both~$y$ and~$y'$.
Now $G[z,z',x,y',y]$ is a~$\overline{2P_1+P_3}$, a contradiction.
We conclude that every vertex in~$A^{*G}_1$ has the same neighbourhood in~$C$.
This implies that every vertex of~$C$ is either complete or anti-complete to~$A^{*G}_1$.
If a vertex of~$C$ is anti-complete to~$A^{*G}_1$, then it is anti-complete to~$A^G_1$, and so it lies in~$B^G$.

Let~$D^{*G}$ be the set of vertices in~$D^G$ that are in components of~$G[D^G]$ that have at most three vertices.
Then every vertex of $D^G \setminus D^{*G}$ is complete or anti-complete to~$A^{*G}_1$ and anti-complete to $A^G_1 \setminus A^{*G}_1$.

Now let $G'=G[D^{*G} \cup L^G \setminus (A^G_1 \setminus A^{*G}_1)]$ and note that this graph is uniquely defined by~$G$ and~$K^G$.
Then~$G'[D^{*G}]$ is $K_4$-free, so $G'[D^{*G}\cup A^{*G}_1]$ is $K_5$-free, so every induced~$K_5$ in~$G'$ is entirely contained in~$L^G\setminus (A^G_1 \setminus A^{*G}_1)$.
Furthermore, since $p \geq 5$, every vertex in~$L^G \setminus (A^G_1 \setminus A^{*G}_1)$ is contained in an induced~$K_5$ in~$G'$.
Therefore every isomorphism~$q$ from~$G'$ to~$H'$ satisfies $q(L^G \setminus (A^G_1 \setminus A^{*G}_1))=L^H \setminus (A^H_1 \setminus A^{*H}_1)$.
Therefore a bijection $f:V(G) \to V(H)$ is an isomorphism from~$G$ to~$H$ such that $f(L^G)=L^H$ if and only if all of the following hold:
\begin{enumerate}
\item\label{prop:1}The restriction of~$f$ to~$V(G')$ is an isomorphism from~$G'$ to~$H'$ such that $f(A^{*G}_1)=A^{*H}_1$.
\item\label{prop:2}$f(A^G_1 \setminus A^{*G}_1)=A^H_1 \setminus A^{*H}_1$.
\item\label{prop:3}For every component~$C$ of~$G[D^G]$ with at least four vertices, $f(C)$ is a component of~$H[D^H]$ on the same number of vertices and $|C \cap B^G|=|f(C) \cap B^H|$.
\end{enumerate}
It is therefore sufficient to test whether there is a bijection from~$G$ to~$H$ with the above properties.
Note that these properties are defined on pairwise disjoint vertex sets, and the edges in~$G$ and~$H$ between these sets are completely determined by the definition of the sets.
Thus it is sufficient to independently test whether there are bijections satisfying each of these properties.
If~$D^{*G}$ is empty, then~$G'$ is a complete multipartite graph, so we can easily test if Property~\ref{prop:1} holds in this case.
Otherwise, since~$A^G_j$ has no neighbours outside~$L^G$ for $j \in \{2,\ldots,p\}$, every isomorphism from~$G'$ to~$H'$ satisfies $f(A^{*G}_1)=A^{*H}_1$, 
so it is sufficient to test if~$G'$ and~$H'$ are isomorphic, and we can do this by applying Case~\ref{case1:D-2P2-free} or Case~\ref{case1:B-K4-free}.
The sets $A^G_1 \setminus A^{*G}_1$ and $A^H_1 \setminus A^{*H}_1$ consist of at most one vertex, so we can test if Property~\ref{prop:2} can be satisfied in polynomial time.
To satisfy Property~\ref{prop:3}, we only need to check whether there is a bijection~$q$ from the components of $G[D^{*G}\setminus D^G]$ to the components of $H[D^{*H}\setminus D^H]$ such that $|q(C)|=|C|$ and $|q(C) \cap B^H|=|C \cap B^G|$ for every component of $G[D^{*G}\setminus D^G]$ and this can clearly be done in polynomial time.
This completes the proof of Case~\ref{case1:B-K4-non-free}.\qed
\end{proof}

\section{New \GI-complete Results}\label{sec:GI-complete}

We state  Theorems~\ref{thm:diamond-2P3-hard},~\ref{thm:diamond-P6-hard} and~\ref{thm:gem-P1+2P2-hard}, which establish that {\sc Graph Isomorphism} is \GI-complete on $(\diamondgraph,2P_3)$-free, $(\diamondgraph,P_6)$-free and $(\gem,P_1+\nobreak 2P_2)$-free graphs, respectively (see \figurename~\ref{fig:forb-hard}).
The complexity of {\sc Graph Isomorphism} on $(\overline{2P_1+P_3},2P_3)$-free graphs and $(\gem,P_6)$-free graphs was previously unknown, but since these classes contain the classes of $(\diamondgraph,2P_3)$-free graphs and $(\diamondgraph,P_6)$-free graphs, respectively, Theorems~\ref{thm:diamond-2P3-hard} and~\ref{thm:diamond-P6-hard}, respectively, imply that {\sc Graph Isomorphism} is also \GI-complete on these classes.
In Theorems~\ref{thm:diamond-2P3-hard} and~\ref{thm:diamond-P6-hard}, \GI-completeness follows from the fact that the constructions used in our proofs fall into the framework of so-called simple path encodings (see~\cite{Sc17}).
For brevity, we do not explain this general notion here, but instead include direct proofs of \GI-completeness for both cases.
The construction used in the proof of Theorem~\ref{thm:gem-P1+2P2-hard} does not fall into this framework and we give a direct proof of \GI-completeness in this case.

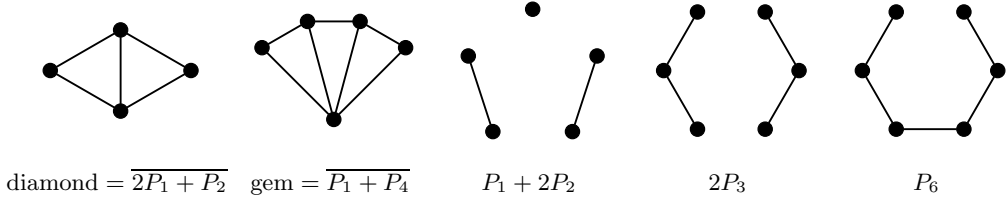
\begin{figure}
\begin{center}
\begin{tabular}{ccccc}
\begin{minipage}{0.187\textwidth}
\begin{center}
\scalebox{0.9}{
\begin{tikzpicture}[scale=1.2]
\GraphInit[vstyle=Simple]
\SetVertexSimple[MinSize=6pt]
\Vertex[x=0,y=0.5]{a}
\Vertex[x=0.866025403784438646763723,y=0]{b}
\Vertex[x=0,y=-0.5]{c}
\Vertex[x=-0.866025403784438646763723,y=0]{d}
\Edges(a,b,c,d,a,c)
\end{tikzpicture}}
\end{center}
\end{minipage}
&
\begin{minipage}{0.187\textwidth}
\begin{center}
\scalebox{0.9}{
\begin{tikzpicture}[scale=1.5,rotate=90]
\GraphInit[vstyle=Simple]
\SetVertexSimple[MinSize=6pt]
\Vertex[x=0,y=0]{x}
\Vertex[a=-45,d=1]{a}
\Vertex[a=-15,d=1]{b}
\Vertex[a=15,d=1]{c}
\Vertex[a=45,d=1]{d}
\Edges(a,x,b)
\Edges(c,x,d)
\Edges(a,b,c,d)
\end{tikzpicture}}
\end{center}
\end{minipage}
&
\begin{minipage}{0.187\textwidth}
\begin{center}
\scalebox{0.9}{
\begin{tikzpicture}[scale=1,rotate=90]
\GraphInit[vstyle=Simple]
\SetVertexSimple[MinSize=6pt]
\Vertices{circle}{a,b,c,d,e}
\Edges(b,c)
\Edges(d,e)
\end{tikzpicture}}
\end{center}
\end{minipage}
&
\begin{minipage}{0.187\textwidth}
\begin{center}
\scalebox{0.9}{
\begin{tikzpicture}[scale=1]
\GraphInit[vstyle=Simple]
\SetVertexSimple[MinSize=6pt]
\Vertices{circle}{a,b,c,d,e,f}
\Edges(c,d,e)
\Edges(f,a,b)
\end{tikzpicture}}
\end{center}
\end{minipage}
&
\begin{minipage}{0.187\textwidth}
\begin{center}
\scalebox{0.9}{
\begin{tikzpicture}[scale=1]
\GraphInit[vstyle=Simple]
\SetVertexSimple[MinSize=6pt]
\Vertices{circle}{a,b,c,d,e,f}
\Edges(c,d,e,f,a,b)
\end{tikzpicture}}
\end{center}
\end{minipage}
\\
\\
$\diamondgraph=\overline{2P_1+P_2}$ & $\gem=\overline{P_1+P_4}$ & $P_1+\nobreak 2P_2$ & $2P_3$ & $P_6$
\end{tabular}
\end{center}
\caption{\label{fig:forb-hard}Forbidden induced subgraphs from Theorems~\ref{thm:diamond-2P3-hard},~\ref{thm:diamond-P6-hard} and~\ref{thm:gem-P1+2P2-hard}.}
\end{figure}

\begin{sloppypar}
\begin{theorem}\label{thm:diamond-2P3-hard}
{\sc Graph Isomorphism} is \GI-complete on $(\diamondgraph,2P_3)$-free graphs.
\end{theorem}
\end{sloppypar}

\begin{proof}
Let~$G$ be a graph.
We construct a graph~$q(G)$ as follows:
\begin{enumerate}
\item Create a clique with vertex set $A^G=V(G)$.
\item For every edge $vw \in E(G)$, add vertices~$v_w$ and~$w_v$ and edges $vv_w,v_ww_v$ and~$w_vw$.
Let~$B^G$ be the set of vertices added in this step.
\end{enumerate}
Note that every vertex in~$B^G$ has exactly two neighbours in~$q(G)$ and that these neighbours are non-adjacent.
Therefore no induced~$K_3$ in~$q(G)$ contains a vertex of~$B^G$.
Also note that $|A^G|=|V(G)|$ and $|B^G|=2|E(G)|$.

We claim that~$q(G)$ is $(\diamondgraph,2P_3)$-free for every graph~$G$.
First suppose, for contradiction, that the $\diamondgraph$ is an induced subgraph of~$q(G)$.
Since no vertex in~$B^G$ is in an induced~$K_3$ in~$q(G)$, it follows that no vertex of this $\diamondgraph$ can be in~$B^G$.
This is a contradiction, since~$A^G$ is a clique.
Therefore~$q(G)$ is $\diamondgraph$-free.
Now suppose, for contradiction, that~$2P_3$ is an induced subgraph of~$q(G)$.
Since $q(G)[B^G]$ is a disjoint union of~$P_2$'s, every~$P_3$ in~$q(G)$ must contain at least one vertex in~$A^G$.
Therefore, the two components of the~$2P_3$ must each contain a vertex of~$A^G$ and so there must be two non-adjacent vertices in~$A^G$.
Since~$A^G$ is a clique, this is a contradiction.
Therefore~$q(G)$ is $2P_3$-free.

Given two graphs~$G$ and~$H$, we claim that~$G$ is isomorphic to~$H$ if and only if~$q(G)$ is isomorphic to~$q(H)$.
Clearly, if~$G$ is isomorphic to~$H$, then~$q(G)$ is isomorphic to~$q(H)$.
Now suppose that there is an isomorphism~$f$ from~$q(G)$ to~$q(H)$.
Let us show that this implies~$G$ is isomorphic to~$H$.
If~$G$ or~$H$ contains at most two vertices, then this can be verified by inspection, so we may assume $|V(G)|,|V(H)| \geq 3$.
It follows that every vertex of~$A^G$ (resp.~$A^H$) is in an induced~$K_3$ in~$q(G)$ (resp.~$q(H)$).
Since no vertex of~$B^G$ (resp.~$B^H$) is in an induced~$K_3$ in~$q(G)$ (resp.~$q(H)$), it follows that $f(A^G)=A^H$ and $f(B^G)=B^H$.
Now two vertices~$v$ and~$w$ in~$G$ are adjacent if and only if~$v$ and~$w$ are connected in~$q(G)$ via a path of vertices in~$B^G$ if and only if~$f(v)$ and~$f(w)$ are connected in~$q(H)$ via a path of vertices in~$B^H$ if and only if~$f(v)$ and~$f(w)$ are adjacent in~$H$.
Therefore~$G$ is isomorphic to~$H$.
This completes the proof.\qed
\end{proof}

\begin{theorem}\label{thm:diamond-P6-hard}
{\sc Graph Isomorphism} is \GI-complete on $(\diamondgraph,P_6)$-free graphs.
\end{theorem}

\begin{proof}
Let~$G$ be a graph.
We construct a graph~$q(G)$ as follows:
\begin{enumerate}
\item Create an independent set with vertex set $A^G=V(G)$.
\item Create an independent set with vertex set $C^G=E(G)$.
\item Add every possible edge between~$A^G$ and~$C^G$.
\begin{sloppypar}
\item For every edge $e=vw \in E(G)$, add vertices~$v_w$ and~$w_v$ and edges $vv_w,v_we,ew_v$ and~$w_vw$ (note that $e \in C^G$).
Let~$B^G$ be the set of vertices added in this step.
\end{sloppypar}
\end{enumerate}
Note that every vertex in~$B^G$ has exactly two neighbours in~$q(G)$ and these neighbours are adjacent.
Furthermore, note that $A^G,B^G$ and~$C^G$ are independent sets with $|A^G|=|V(G)|$ and $|B^G|=2|E(G)|=2|C^G|$.

We claim that~$q(G)$ is $(\diamondgraph,P_6)$-free for every graph~$G$.
First suppose, for contradiction, that the $\diamondgraph$ is an induced subgraph of~$q(G)$.
Since $A^G,B^G$ and~$C^G$ are independent sets, every induced~$K_3$ in~$q(G)$ must have exactly one vertex from each of these sets.
Therefore, since the vertices of~$B^G$ have degree-$2$ in~$q(G)$,
one of the degree-$3$ vertices of the $\diamondgraph$ must be in~$A^G$ and the other in~$C^G$,
and so both degree-$2$ vertices of the $\diamondgraph$ must be in~$B^G$.
However no pair of vertices in~$B^G$ has the same neighbour in~$A^G$ and the same neighbour in~$C^G$, a contradiction.
We conclude that~$q(G)$ is $\diamondgraph$-free.
Now suppose, for contradiction, that~$P_6$ is an induced subgraph of~$q(G)$.
Since the two neighbours of every vertex in~$B^G$ are adjacent, the internal vertices of the~$P_6$ cannot lie in~$B^G$.
Therefore $q(G)[A^G \cup C^G]$ contains an induced~$P_4$.
Since $q(G)[A^G \cup C^G]$ is a complete bipartite graph, it is $P_4$-free.
This contradiction implies that~$q(G)$ is $P_6$-free.

Now let~$G$ and~$H$ be graphs.
Let~$G^*$ and~$H^*$ be the graphs obtained from~$G$ and~$H$, respectively, by adding four pairwise adjacent vertices that are adjacent to every vertex of~$G$ and~$H$, respectively.
Given two graphs~$G$ and~$H$, we claim that~$G$ is isomorphic to~$H$ if and only if~$q(G^*)$ is isomorphic to~$q(H^*)$.
Clearly if~$G$ is isomorphic to~$H$, then~$q(G^*)$ is isomorphic to~$q(H^*)$.
Furthermore, $G$ is isomorphic to~$H$ if and only if~$G^*$ is isomorphic to~$H^*$.
Now suppose that there is an isomorphism~$f$ from~$q(G^*)$ to~$q(H^*)$.
It suffices to show that~$G^*$ is isomorphic to~$H^*$.
Note that $|V(G^*)| \geq 4$ and $|E(G^*)| \geq 6$ by construction.
Thus every vertex in $A^{G^*} \cup C^{G^*}$ has degree greater than~$2$ in~$q(G^*)$.
Since every vertex in~$B^{G^*}$ has degree~$2$ in~$q(G^*)$, it follows that a vertex of~$q(G^*)$ has degree exactly~$2$ if and only if it is in~$B^{G^*}$.
Similarly, a vertex of~$q(H^*)$ has degree~$2$ if and only if it is in~$B^{H^*}$.
Therefore $f(B^{G^*})=B^{H^*}$, and so $|B^{G^*}|=|B^{H^*}|$.
Since $B^{G^*}=2|E(G^*)|$ and $|B^{H^*}|=2|E(H^*)|$, it follows that $|E(G^*)|=|E(H^*)|$.
Since~$q(G^*)$ has $|V(G^*)|+3|E(G^*)|$ vertices and~$q(H^*)$ has $|V(H^*)|+3|E(H^*)|$ vertices, it follows that $|V(G^*)|=|V(H^*)|$.
Now $q(G^*) \setminus B^{G^*}$ is a complete bipartite graph with parts of size~$|V(G^*)|$ and~$|E(G^*)|$, respectively.
Since we obtained~$G^*$ from~$G$ by adding four vertices that are complete to every other vertex of~$G^*$, it follows that $|E(G^*)| = |E(G)|+4(|V(G^*)|-4)+6\geq 3(|V(G^*)|-4) + |V(G^*)|+2 > |V(G^*)|$.
We conclude that $f(A^{G^*})=f(A^{H^*})$ and $f(C^{G^*})=C^{H^*}$.
Now two vertices~$v$ and~$w$ in~$G^*$ are adjacent if and only if~$v$ and~$w$ are connected in~$q(G^*)$ via a path of vertices in $B^{G^*},C^{G^*}$ and~$B^{G^*}$, respectively if and only if~$f(v)$ and~$f(w)$ are connected in~$q(H^*)$ via a path of vertices in~$B^{H^*},C^{H^*}$ and~$B^{H^*}$, respectively, if and only if~$f(v)$ and~$f(w)$ are adjacent in~$H^*$.
Therefore~$G^*$ is isomorphic to~$H^*$.

Combining the above with the fact that~$q(G^*)$ and~$q(H^*)$ are $(\diamondgraph,P_6)$-free shows that {\sc Graph Isomorphism} is \GI-complete on $(\diamondgraph,P_6)$-free graphs.\qed
\end{proof}

\begin{theorem}\label{thm:gem-P1+2P2-hard}
{\sc Graph Isomorphism} is \GI-complete on $(\gem,P_1+\nobreak 2P_2)$-free graphs.
Furthermore, $(\gem,P_1+\nobreak 2P_2)$-free graphs have unbounded clique-width.
\end{theorem}

\begin{proof}
Let~$G$ be a graph.
Let $v^G_1,\ldots,v^G_n$ be the vertices of~$G$ and let $e^G_1,\ldots,e^G_m$ be the edges of~$G$.
For the proof of both statements of the theorem, we construct a graph~$q(G)$ from~$G$ as follows:
\begin{enumerate}
\item Create a complete multipartite graph with partition $(A^G_1,\ldots,A^G_n)$, where $|A^G_i|=d_G(v^G_i)$ for $i \in \{1,\ldots,n\}$ and let $A^G=\bigcup A^G_i$.
\item Create a complete multipartite graph with partition $(B^G_1,\ldots,B^G_m)$, where $|B^G_i|=2$ for $i \in \{1,\ldots,m\}$ and let $B^G=\bigcup B^G_i$.
\item Take the disjoint union of the two graphs above, then for each edge $e^G_i=v^G_{i_1}v^G_{i_2}$ in~$G$ in turn, add an edge from one vertex of~$B^G_i$ to a vertex of~$A^G_{i_1}$ and an edge from the other vertex of~$B^G_i$ to a vertex of~$A^G_{i_2}$.
Do this in such a way that the edges added between~$A^G$ and~$B^G$ form a perfect matching.
\end{enumerate}

We claim that~$q(G)$ is $(\gem,P_1+\nobreak 2P_2)$-free.
Since~$q(G)[A^G]$ and~$q(G)[B^G]$ are complete multipartite graphs, they must both be $(P_1+\nobreak P_2)$-free, so every induced $P_1+\nobreak P_2$ in~$q(G)$ must contain at least one vertex in~$A^G$ and at least one vertex in~$B^G$.
Suppose, for contradiction, that the $\gem$ is an induced subgraph of~$q(G)$.
Let $X \in \{A^G,B^G\}$ be the set that contains the dominating vertex~$v$ of the~$\gem$ and let~$Y$ be the other set.
Since $\gem-v$ is isomorphic to~$P_4$, which contains an induced~$P_1+\nobreak P_2$, at least one vertex~$w$ of the~$\gem$ must be in~$Y$.
Since~$v$ has only one neighbour in~$Y$, all other vertices of the~$\gem$ must be in~$X$.
However, $w$ has only one neighbour in~$X$, but at least two neighbours in the~$\gem$.
This contradiction shows that~$q(G)$ is indeed $\gem$-free.
Now suppose, for contradiction, that~$P_1+\nobreak 2P_2$ is an induced subgraph of~$q(G)$.
First suppose that one of the~$P_2$'s in this~$P_1+\nobreak 2P_2$ either has both vertices in~$A^G$ or both vertices in~$B^G$; let $X\in \{A^G,B^G\}$ be the set that contains this~$P_2$ and let~$Y$ be the other set.
Then since~$q(G)[X]$ is $(P_1+\nobreak P_2)$-free, the remaining three vertices of the~$P_1+\nobreak 2P_2$ must be in~$Y$.
This means that~$q(G)[Y]$ contains~$P_1+\nobreak P_2$ as an induced subgraph.
This contradiction means that each of the~$P_2$'s in the~$P_1+\nobreak 2P_2$ must have exactly one vertex in~$A^G$ and exactly one vertex in~$B^G$.
Therefore there must be non-adjacent vertices $x_1,x_2 \in A^G$ and non-adjacent vertices $y_1,y_2 \in B^G$ such that~$x_i$ is adjacent to~$y_j$ if and only if $i=j$.
Therefore~$x_1$ and~$x_2$ must be in the same set~$A^G_i$ and~$y_1$ and~$y_2$ must be in the same set~$B^G_j$.
This is a contradiction as the two vertices in~$B^G_j$ cannot both have neighbours in the same set~$A^G_i$.
Therefore~$q(G)$ is indeed $(\gem,P_1+\nobreak 2P_2)$-free.

\medskip
\noindent
We are now ready to prove that {\sc Graph Isomorphism} is \GI-complete on $(\gem,P_1+\nobreak 2P_2)$-free graphs.
Now let~$G$ and~$H$ be graphs.
Let~$G^*$ and~$H^*$ be the graphs obtained from~$G$ and~$H$, respectively, by adding four pairwise adjacent vertices that are adjacent to every vertex of~$G$ and~$H$, respectively.
Note that every vertex of~$G^*$ and~$H^*$ has degree at least~$3$.
We claim that~$G$ is isomorphic to~$H$ if and only if~$q(G^*)$ is isomorphic to~$q(H^*)$.
Clearly if~$G$ is isomorphic to~$H$, then~$q(G^*)$ is isomorphic to~$q(H^*)$.
Furthermore, $G$ is isomorphic to~$H$ if and only if~$G^*$ is isomorphic to~$H^*$.
Now suppose that there is an isomorphism~$f$ from~$q(G^*)$ to~$q(H^*)$.
It suffices to show that~$G^*$ is isomorphic to~$H^*$.
Note that~$q(G^*)[A^{G^*}]$ and~$q(G^*)[B^{G^*}]$ each contain an induced~$K_3$ (but there is no~$K_3$ in~$q(G^*)$ with vertices in both~$A^{G^*}$ and~$B^{G^*}$).
Furthermore, given such a~$K_3$ in~$q(G^*)[A^{G^*}]$ (resp.~$q(G^*)[B^{G^*}]$), a vertex is in~$A^{G^*}$ (resp.~$B^{G^*}$) if and only if it has at least two neighbours in this~$K_3$, so either $f(A^{G^*})=A^{H^*}$ and $f(B^{G^*})=B^{H^*}$ or $f(A^{G^*})=B^{H^*}$ and $f(B^{G^*})=A^{H^*}$.
Since~$q(G^*)[A^{G^*}]$ contains an induced~$3P_1$, but~$q(G)[B^{G^*}]$ does not, it follows that $f(A^{G^*})=A^{H^*}$ and $f(B^{G^*})=B^{H^*}$.
Furthermore, this implies that for all $i \in \{1,\ldots,n\}$, $f(A^{G^*}_i)=A^{H^*}_j$ for some $j \in \{1,\ldots,n\}$ with $|A^{H^*}_j|=|A^{G^*}_i|$ and for all $i \in \{1,\ldots,m\}$, $f(B^{G^*}_i)=B^{H^*}_j$ for some $j \in \{1,\ldots,m\}$.
Now two vertices~$v^{G^*}_i$ and~$v^{G^*}_j$ in~$G^*$ are adjacent if and only if there is a $k \in \{1,\ldots,m\}$ such that there are edges in~$q(G^*)$ from~$B^{G^*}_k$ to both~$A^{G^*}_i$ and~$A^{G^*}_j$ if and only if there is a $k \in \{1,\ldots,m\}$ such that there are edges in~$q(H^*)$ from~$f(B^{G^*}_k)$ to both~$f(A^{G^*}_i)$ and~$f(A^{G^*}_j)$ if and only if~$v^{H^*}_{i'}$ and~$v^{H^*}_{j'}$ are adjacent where $f(A^{G^*}_i)=A^{H^*}_{i'}$ and $f(A^{G^*}_j)=A^{H^*}_{j'}$.
Therefore~$G^*$ is isomorphic to~$H^*$.

Combining the above with the fact that $q(G^*)$ and~$q(H^*)$ are $(\gem,P_1+\nobreak 2P_2)$-free shows that {\sc Graph Isomorphism} is \GI-complete on $(\gem,P_1+\nobreak 2P_2)$-free graphs.

\medskip
\noindent
We now prove that  the class of $(\gem,P_1+\nobreak 2P_2)$-free graphs has unbounded clique-width.
Let~$H_n$ be the $n \times n$ grid (see also \figurename~\ref{fig:grid}).
We claim that the set of graphs $\{q(H_n) \;|\; n \in \mathbb{N}\}$ has unbounded clique-width and note that we have shown that every graph in this set is $(\gem,P_1+\nobreak 2P_2)$-free.
Let~$H_n'$ be the graph obtained from~$q(H_n)$ by complementing~$A^{H_n}$ and complementing~$B^{H_n}$ (see also \figurename~\ref{fig:grid}).
By Fact~\ref{fact:comp}, it is sufficient to show that the set of graphs $\{H_n' \;|\; n \in \mathbb{N}\}$ has unbounded clique-width.
We now partition~$V(H_n')$ into sets~$V_{i,j}$ for $i,j \in \{1,\ldots,n\}$ as follows.
For $i,j \in \{1,\ldots,n\}$ let~$V_{i,j}$ consist of the vertices in the set~$A^{H_n}_k$ that correspond to the vertex in the $i$th row and $j$th column of~$H_n$, along with the vertices in~$B^{H_n}$ that have a neighbour in~$A^{H_n}_k$.
Note that every vertex of~$H_n'$ is in exactly one set~$V_{i,j}$, so these sets form a partition of~$V(H_n')$.
Furthermore $H_n'[\cup^n_{j=1}V_{i,j}]$ is connected for all $i\geq 1$, $H_n'[\cup^n_{i=1}V_{i,j}]$ is connected for all $j\geq 1$, and for $i,j,k,\ell\geq 1$, if a vertex of~$V_{i,j}$ is adjacent to a vertex of~$V_{k,\ell}$, then $|k-i|\leq 1$ and $|\ell-j| \leq 1$.
Applying Lemma~\ref{lem:generalunbounded} with $m=1$ we find that~$H_n'$ has clique-width at least $\lfloor\frac{n-1}{2}\rfloor+\nobreak 1$.
This completes the proof.\qed
\end{proof}

\begin{figure}
\begin{center}
\begin{tabular}{cc}
\begin{minipage}{0.49\textwidth}
\begin{center}
\scalebox{1.0}{
\begin{tikzpicture}[scale=2, every node/.style={scale=0.7}]
\GraphInit[vstyle=Simple]
\SetVertexSimple[MinSize=6pt]
\Vertex[x=1,y=1]{v11}
\Vertex[x=2,y=1]{v21}
\Vertex[x=3,y=1]{v31}
\Vertex[x=4,y=1]{v41}

\Vertex[x=1,y=2]{v12}
\Vertex[x=2,y=2]{v22}
\Vertex[x=3,y=2]{v32}
\Vertex[x=4,y=2]{v42}

\Vertex[x=1,y=3]{v13}
\Vertex[x=2,y=3]{v23}
\Vertex[x=3,y=3]{v33}
\Vertex[x=4,y=3]{v43}

\Vertex[x=1,y=4]{v14}
\Vertex[x=2,y=4]{v24}
\Vertex[x=3,y=4]{v34}
\Vertex[x=4,y=4]{v44}

\Edges(v11,v12,v13,v14)
\Edges(v21,v22,v23,v24)
\Edges(v31,v32,v33,v34)
\Edges(v41,v42,v43,v44)

\Edges(v11,v21,v31,v41)
\Edges(v12,v22,v32,v42)
\Edges(v13,v23,v33,v43)
\Edges(v14,v24,v34,v44)
\end{tikzpicture}}
\end{center}
\end{minipage}
&
\begin{minipage}{0.49\textwidth}
\begin{center}
\scalebox{1.0}{
\begin{tikzpicture}[scale=2, every node/.style={scale=0.7}]
\GraphInit[vstyle=Simple]
\SetVertexSimple[MinSize=6pt]
\Vertex[x=1,y=1+0.2]{v11n}
\Vertex[x=2,y=1+0.2]{v21n}
\Vertex[x=3,y=1+0.2]{v31n}
\Vertex[x=4,y=1+0.2]{v41n}

\Vertex[x=1,y=2+0.2]{v12n}
\Vertex[x=2,y=2+0.2]{v22n}
\Vertex[x=3,y=2+0.2]{v32n}
\Vertex[x=4,y=2+0.2]{v42n}

\Vertex[x=1,y=3+0.2]{v13n}
\Vertex[x=2,y=3+0.2]{v23n}
\Vertex[x=3,y=3+0.2]{v33n}
\Vertex[x=4,y=3+0.2]{v43n}

\Vertex[x=1,y=2-0.2]{v12s}
\Vertex[x=2,y=2-0.2]{v22s}
\Vertex[x=3,y=2-0.2]{v32s}
\Vertex[x=4,y=2-0.2]{v42s}

\Vertex[x=1,y=3-0.2]{v13s}
\Vertex[x=2,y=3-0.2]{v23s}
\Vertex[x=3,y=3-0.2]{v33s}
\Vertex[x=4,y=3-0.2]{v43s}

\Vertex[x=1,y=4-0.2]{v14s}
\Vertex[x=2,y=4-0.2]{v24s}
\Vertex[x=3,y=4-0.2]{v34s}
\Vertex[x=4,y=4-0.2]{v44s}

\Vertex[x=1+0.2,y=1]{v11e}
\Vertex[x=2+0.2,y=1]{v21e}
\Vertex[x=3+0.2,y=1]{v31e}

\Vertex[x=1+0.2,y=2]{v12e}
\Vertex[x=2+0.2,y=2]{v22e}
\Vertex[x=3+0.2,y=2]{v32e}

\Vertex[x=1+0.2,y=3]{v13e}
\Vertex[x=2+0.2,y=3]{v23e}
\Vertex[x=3+0.2,y=3]{v33e}

\Vertex[x=1+0.2,y=4]{v14e}
\Vertex[x=2+0.2,y=4]{v24e}
\Vertex[x=3+0.2,y=4]{v34e}

\Vertex[x=2-0.2,y=1]{v21w}
\Vertex[x=3-0.2,y=1]{v31w}
\Vertex[x=4-0.2,y=1]{v41w}

\Vertex[x=2-0.2,y=2]{v22w}
\Vertex[x=3-0.2,y=2]{v32w}
\Vertex[x=4-0.2,y=2]{v42w}

\Vertex[x=2-0.2,y=3]{v23w}
\Vertex[x=3-0.2,y=3]{v33w}
\Vertex[x=4-0.2,y=3]{v43w}

\Vertex[x=2-0.2,y=4]{v24w}
\Vertex[x=3-0.2,y=4]{v34w}
\Vertex[x=4-0.2,y=4]{v44w}

\SetVertexSimple[MinSize=6pt,FillColor=white]
\Vertex[x=1,y=1+0.4]{v11nn}
\Vertex[x=2,y=1+0.4]{v21nn}
\Vertex[x=3,y=1+0.4]{v31nn}
\Vertex[x=4,y=1+0.4]{v41nn}

\Vertex[x=1,y=2+0.4]{v12nn}
\Vertex[x=2,y=2+0.4]{v22nn}
\Vertex[x=3,y=2+0.4]{v32nn}
\Vertex[x=4,y=2+0.4]{v42nn}

\Vertex[x=1,y=3+0.4]{v13nn}
\Vertex[x=2,y=3+0.4]{v23nn}
\Vertex[x=3,y=3+0.4]{v33nn}
\Vertex[x=4,y=3+0.4]{v43nn}

\Vertex[x=1,y=2-0.4]{v12ss}
\Vertex[x=2,y=2-0.4]{v22ss}
\Vertex[x=3,y=2-0.4]{v32ss}
\Vertex[x=4,y=2-0.4]{v42ss}

\Vertex[x=1,y=3-0.4]{v13ss}
\Vertex[x=2,y=3-0.4]{v23ss}
\Vertex[x=3,y=3-0.4]{v33ss}
\Vertex[x=4,y=3-0.4]{v43ss}

\Vertex[x=1,y=4-0.4]{v14ss}
\Vertex[x=2,y=4-0.4]{v24ss}
\Vertex[x=3,y=4-0.4]{v34ss}
\Vertex[x=4,y=4-0.4]{v44ss}

\Vertex[x=1+0.4,y=1]{v11ee}
\Vertex[x=2+0.4,y=1]{v21ee}
\Vertex[x=3+0.4,y=1]{v31ee}

\Vertex[x=1+0.4,y=2]{v12ee}
\Vertex[x=2+0.4,y=2]{v22ee}
\Vertex[x=3+0.4,y=2]{v32ee}

\Vertex[x=1+0.4,y=3]{v13ee}
\Vertex[x=2+0.4,y=3]{v23ee}
\Vertex[x=3+0.4,y=3]{v33ee}

\Vertex[x=1+0.4,y=4]{v14ee}
\Vertex[x=2+0.4,y=4]{v24ee}
\Vertex[x=3+0.4,y=4]{v34ee}

\Vertex[x=2-0.4,y=1]{v21ww}
\Vertex[x=3-0.4,y=1]{v31ww}
\Vertex[x=4-0.4,y=1]{v41ww}

\Vertex[x=2-0.4,y=2]{v22ww}
\Vertex[x=3-0.4,y=2]{v32ww}
\Vertex[x=4-0.4,y=2]{v42ww}

\Vertex[x=2-0.4,y=3]{v23ww}
\Vertex[x=3-0.4,y=3]{v33ww}
\Vertex[x=4-0.4,y=3]{v43ww}

\Vertex[x=2-0.4,y=4]{v24ww}
\Vertex[x=3-0.4,y=4]{v34ww}
\Vertex[x=4-0.4,y=4]{v44ww}

\Edges(v11e,v11ee,v21ww,v21w,v21e,v21ee,v31ww,v31w,v31e,v31ee,v41ww,v41w)
\Edges(v12e,v12ee,v22ww,v22w,v22e,v22ee,v32ww,v32w,v32e,v32ee,v42ww,v42w)
\Edges(v13e,v13ee,v23ww,v23w,v23e,v23ee,v33ww,v33w,v33e,v33ee,v43ww,v43w)
\Edges(v14e,v14ee,v24ww,v24w,v24e,v24ee,v34ww,v34w,v34e,v34ee,v44ww,v44w)

\Edges(v11n,v11nn,v12ss,v12s,v12n,v12nn,v13ss,v13s,v13n,v13nn,v14ss,v14s)
\Edges(v21n,v21nn,v22ss,v22s,v22n,v22nn,v23ss,v23s,v23n,v23nn,v24ss,v24s)
\Edges(v31n,v31nn,v32ss,v32s,v32n,v32nn,v33ss,v33s,v33n,v33nn,v34ss,v34s)
\Edges(v41n,v41nn,v42ss,v42s,v42n,v42nn,v43ss,v43s,v43n,v43nn,v44ss,v44s)

\Edges(v11n,v11e)
\Edges(v12n,v12e,v12s)
\Edges(v13n,v13e,v13s)
\Edges(v14e,v14s)

\Edges(v21w,v21n,v21e)
\Edges(v22n,v22e,v22s,v22w,v22n)
\Edges(v23n,v23e,v23s,v23w,v23n)
\Edges(v24e,v24s,v24w)

\Edges(v31w,v31n,v31e)
\Edges(v32n,v32e,v32s,v32w,v32n)
\Edges(v33n,v33e,v33s,v33w,v33n)
\Edges(v34e,v34s,v34w)

\Edges(v41n,v41w)
\Edges(v42s,v42w,v42n)
\Edges(v43s,v43w,v43n)
\Edges(v44s,v44w)
\end{tikzpicture}}
\end{center}
\end{minipage}
\end{tabular}
\end{center}
\caption{\label{fig:grid}The $n\times n$ grid~$H_n$ and the graph~$H_n'$, defined in the proof of Theorem~\ref{thm:gem-P1+2P2-hard},
for $n=4$.
In the image of~$H_n'$, the vertices in~$A^{H_n}$ are coloured black and the vertices in~$B^{H_n}$ are coloured white.}
\end{figure}
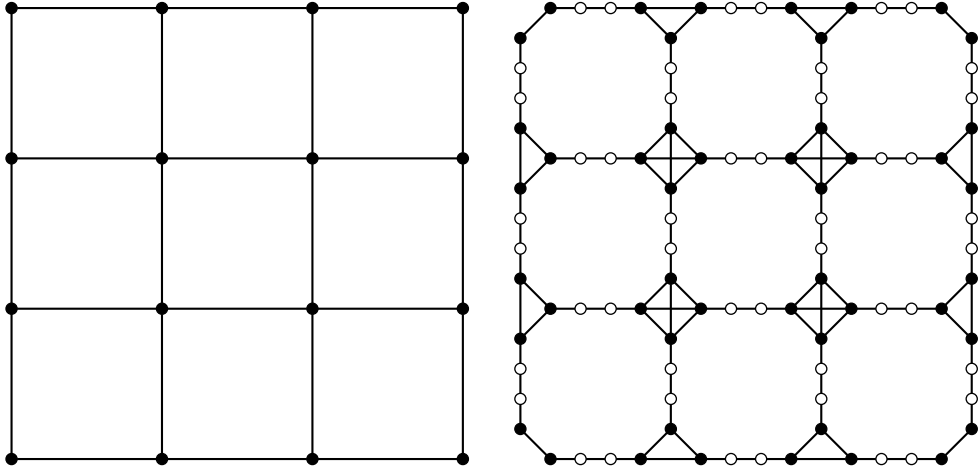

\section{Clique-Width for Hereditary Graph Classes}\label{sec:cw}

The following result (see~\cite{DP16} for a proof), combined with Theorem~\ref{thm:gi-one-graph}, shows that the classifications of the complexity of {\sc Graph Isomorphism} and boundedness of clique-width are analogous for $H$-free graphs.

\begin{theorem}\label{thm:P4-free-bdd-cw}
Let~$H$ be a graph.
The class of $H$-free graphs has bounded clique-width if and only~if $H \ssi P_4$.
\end{theorem}

However, for $(H_1,H_2)$-free graphs, the classifications no longer coincide.
Below, we update the summary theorem and list of open cases from~\cite{DLP17}.
That is, we added the new case solved in Theorem~\ref{thm:gem-P1+2P2-hard} to Theorem~\ref{thm:classification2} (Statement~\ref{thm:classification2:unbdd:gem}) and removed it from Open Problem~\ref{oprob:twographs}.
Given four graphs $H_1,H_2,H_3,H_4$, the classes of $(H_1,H_2)$-free graphs and $(H_3,H_4)$-free graphs are {\em equivalent} if the unordered pair $H_3,H_4$ can be obtained from the unordered pair $H_1,H_2$ by some combination of the operations:
\begin{enumerate}[(i)]
\item complementing both graphs in the pair, and
\item if one of the graphs in the pair is~$K_3$, replacing it with the~$\paw$ or vice versa.
\end{enumerate}
If two classes are equivalent, then one of them has bounded clique-width if and only if the other one does~\cite{DP16}.

\begin{theorem}\label{thm:classification2}
For a class~${\cal G}$ of graphs defined by two forbidden induced subgraphs, the following holds:
\begin{enumerate}
\item ${\cal G}$ has bounded clique-width if it is equivalent to a class of $(H_1,H_2)$-free graphs such that one of the following holds:
\begin{enumerate}[(i)]
\renewcommand{\theenumii}{(\roman{enumii})}
\renewcommand{\labelenumii}{(\roman{enumii})}
\item \label{thm:classification2:bdd:P4} $H_1$ or $H_2 \ssi P_4$
\item \label{thm:classification2:bdd:ramsey} $H_1=K_s$ and $H_2=tP_1$ for some $s,t\geq 1$
\item \label{thm:classification2:bdd:P_1+P_3} $H_1 \ssi \paw$ and $H_2 \ssi K_{1,3}+\nobreak 3P_1,\; K_{1,3}+\nobreak P_2,\;\allowbreak P_1+\nobreak P_2+\nobreak P_3,\;\allowbreak P_1+\nobreak P_5,\;\allowbreak P_1+\nobreak S_{1,1,2},\;\allowbreak P_2+\nobreak P_4,\;\allowbreak P_6,\; \allowbreak S_{1,1,3}$ or~$S_{1,2,2}$
\item \label{thm:classification2:bdd:2P_1+P_2} $H_1 \ssi \diamondgraph$ and $H_2\ssi P_1+\nobreak 2P_2,\; 3P_1+\nobreak P_2$ or~$P_2+\nobreak P_3$
\item \label{thm:classification2:bdd:P_1+P_4} $H_1 \ssi \gem$ and $H_2 \ssi P_1+\nobreak P_4$ or~$P_5$
\item \label{thm:classification2:bdd:K_13} $H_1\ssi K_3+\nobreak P_1$ and $H_2 \ssi K_{1,3}$
\item \label{thm:classification2:bdd:2P1_P3} $H_1\ssi \overline{2P_1+\nobreak P_3}$ and $H_2\ssi 2P_1+\nobreak P_3$.
\end{enumerate}
\item ${\cal G}$ has unbounded clique-width if it is equivalent to a class of $(H_1,H_2)$-free graphs such that one of the following holds:
\begin{enumerate}[(i)]
\renewcommand{\theenumii}{(\roman{enumii})}
\renewcommand{\labelenumii}{(\roman{enumii})}
\item \label{thm:classification2:unbdd:not-in-S} $H_1\not\in {\cal S}$ and $H_2 \not \in {\cal S}$
\item \label{thm:classification2:unbdd:not-in-co-S} $H_1\notin \overline{{\cal S}}$ and $H_2 \not \in \overline{{\cal S}}$
\item \label{thm:classification2:unbdd:K_13or2P_2} $H_1 \si K_3+\nobreak P_1$ or~$C_4$ and $H_2 \si 4P_1$ or~$2P_2$
\item \label{thm:classification2:unbdd:2P_1+P_2} $H_1 \si \diamondgraph$ and $H_2 \si K_{1,3},\; 5P_1,\; P_2+\nobreak P_4$ or~$P_6$
\item \label{thm:classification2:unbdd:3P_1} $H_1 \si K_3$ and $H_2 \si 2P_1+\nobreak 2P_2,\; 2P_1+\nobreak P_4,\; 4P_1+\nobreak P_2,\; 3P_2$ or~$2P_3$
\item \label{thm:classification2:unbdd:4P_1} $H_1 \si K_4$ and $H_2 \si P_1 +\nobreak P_4$ or~$3P_1+\nobreak P_2$
\item \label{thm:classification2:unbdd:gem} $H_1 \si \gem$ and $H_2 \si P_1+\nobreak 2P_2$.
\end{enumerate}
\end{enumerate}
\end{theorem}

\begin{oproblem}\label{oprob:twographs}
Does the class of $(H_1,H_2)$-free graphs have bounded or unbounded clique-width when:
\begin{enumerate}[(i)]
\renewcommand{\theenumii}{(\roman{enumii})}
\renewcommand{\labelenumii}{(\roman{enumii})}
\item \label{oprob:twographs:3P_1} $H_1=K_3$ and $H_2 \in \{P_1+\nobreak S_{1,1,3},
\allowbreak S_{1,2,3}\}$
\item\label{oprob:twographs:2P_1+P_2} $H_1=\diamondgraph$ and $H_2 \in \{P_1+\nobreak P_2+\nobreak P_3,\allowbreak P_1+\nobreak P_5\}$
\item \label{oprob:twographs:P_1+P_4} $H_1=\gem$ and $H_2=P_2+\nobreak P_3$.
\end{enumerate}
\end{oproblem}

\section{Classifying the Complexity of {\sc Graph Isomorphism} for $(H_1,H_2)$-free Graphs}\label{sec:gi-classification}

\begin{sloppypar}
Recall that given four graphs $H_1,H_2,H_3,H_4$, the classes of $(H_1,H_2)$-free graphs and $(H_3,H_4)$-free graphs are {\em equivalent} if the unordered pair $H_3,H_4$ can be obtained from the unordered pair $H_1,H_2$ by some combination of the operations:
\begin{enumerate}[(i)]
\item complementing both graphs in the pair, and
\item if one of the graphs in the pair is~$K_3$, replacing it with the~$\paw$ or vice versa.
\end{enumerate}
Note that two graphs~$G$ and~$H$ are isomorphic if and only if their complements~$\overline{G}$ and~$\overline{H}$ are isomorphic.
Therefore, for every pair of graphs $H_1,H_2$, the {\sc Graph Isomorphism} problem is polynomial-time solvable or \GI-complete for $(H_1,H_2)$-free graphs if and only if the same is true for $(\overline{H_1},\overline{H_2})$-free graphs.
Since {\sc Graph Isomorphism} can be solved component-wise, and it can easily be solved on complete multipartite graphs in polynomial time, Lemma~\ref{lem:olariu} implies that for every graph~$H_1$, the {\sc Graph Isomorphism} problem is polynomial-time solvable or \GI-complete for $(H_1,K_3)$-free graphs if and only if the same is true for $(H_1,\paw)$-free graphs.
Thus if two classes are equivalent, then the complexity of {\sc Graph Isomorphism} is the same on both of them.
\end{sloppypar}

Here is the summary of known results for the complexity of {\sc Graph Isomorphism} on $(H_1,H_2)$-free graphs (see Section~\ref{sec:prelim} for notation).

\begin{theorem}\label{thm:gi-classification2}
For a class~${\cal G}$ of graphs defined by two forbidden induced subgraphs, the following holds:
\begin{enumerate}
\item {\sc Graph Isomorphism} is solvable in polynomial time on~${\cal G}$ if~${\cal G}$ is equivalent to a class of $(H_1,H_2)$-free graphs such that one of the following holds:
\begin{enumerate}[(i)]
\renewcommand{\theenumii}{(\roman{enumii})}
\renewcommand{\labelenumii}{(\roman{enumii})}
\item \label{thm:gi-classification2:poly:P4}$H_1$ or $H_2 \ssi P_4$
\item \label{thm:gi-classification2:poly:K1t+P1-co-K1t+P1}$\overline{H_1}$ and $H_2 \ssi K_{1,t}+\nobreak P_1$ for some~$t \geq 1$
\item \label{thm:gi-classification2:poly:P3+tP1-coP3+tP1}$\overline{H_1}$ and $H_2 \ssi tP_1+\nobreak P_3$ for some~$t \geq 1$ 
\item \label{thm:gi-classification2:poly:Kt}$H_1 \ssi K_t$ and $H_2 \ssi 2K_{1,t}, K_{1,t}^+$ or~$P_5$ for some $t \geq 1$
\item \label{thm:gi-classification2:poly:paw}$H_1 \ssi \paw$ and $H_2 \ssi P_2+\nobreak P_4, P_6, S_{1,2,2}$ or~$K_{1,t}^{++}+\nobreak P_1$ for some~$t \geq 1$
\item \label{thm:gi-classification2:poly:diamond}$H_1 \ssi \diamondgraph$ and $H_2 \ssi P_1+\nobreak 2P_2$ 
\item \label{thm:gi-classification2:poly:gem}$H_1 \ssi \gem$ and $H_2 \ssi P_1+\nobreak P_4$ or~$P_5$ 
\item \label{thm:gi-classification2:poly:crossed-house}$H_1 \ssi \overline{2P_1+P_3}$ and $H_2 \ssi P_2+\nobreak P_3$ or~$P_5$.
\end{enumerate}
\item {\sc Graph Isomorphism} is \GI-complete on~${\cal G}$ if~${\cal G}$ is equivalent to a class of $(H_1,H_2)$-free graphs such that one of the following holds:
\begin{enumerate}[(i)]
\renewcommand{\theenumii}{(\roman{enumii})}
\renewcommand{\labelenumii}{(\roman{enumii})}
\item \label{thm:gi-classification2:hard:not-path-star-forest}neither~$H_1$ nor~$H_2$ is a path star forest 
\item \label{thm:gi-classification2:hard:not-co-path-star-forest}neither~$\overline{H_1}$ nor~$\overline{H_2}$ is a path star forest 
\item \label{thm:gi-classification2:hard:K3}$H_1\si K_3$ and $H_2\si 2P_1+\nobreak 2P_2,P_1+\nobreak 2P_3,2P_1+\nobreak P_4$ or~$3P_2$
\item \label{thm:gi-classification2:hard:K4}$H_1\si K_4$ and $H_2\si K_{1,4}^{++}, P_1+\nobreak 2P_2$ or~$P_1+\nobreak P_4$
\item \label{thm:gi-classification2:hard:K5}$H_1\si K_5$ and $H_2\si K_{1,3}^{++}$
\item \label{thm:gi-classification2:hard:2P2}$H_1 \si C_4$ and $H_2 \si K_{1,3}, 3P_1+\nobreak P_2$ or~$2P_2$
\item \label{thm:gi-classification2:hard:2P1+P2}$H_1\si \diamondgraph$ and $H_2 \si K_{1,3}, P_2+\nobreak P_4, 2P_3$ or~$P_6$
\item \label{thm:gi-classification2:hard:gem}$H_1\si \gem$ and $H_2\si P_1+\nobreak 2P_2$.
\end{enumerate}
\end{enumerate}
\end{theorem}

\begin{proof}
In the proof of this theorem we will refer to theorems in a number of other papers, in some cases indicating the value some parameter given therein must take.
Restating and fully explaining all these various theorems in detail is beyond the scope of this paper, but to aid the reader who refers to~\cite{KS12} or~\cite{Sc17}, we note that there $H(a,b,c)$ denotes $K_{1,b}+\nobreak cP_1$ if $a=0$, $H(a,b,c)$ denotes $K_{1,b+1}^{+}+\nobreak cP_1$ if $a=1$, and $H(1,0,b,1)$ denotes $K_{1,b+1}^{++}+\nobreak P_1$.

We first consider the polynomial-time cases.
Statement~\ref{thm:gi-classification2:poly:P4} follows from Theorem~\ref{thm:gi-one-graph}.
Statement~\ref{thm:gi-classification2:poly:K1t+P1-co-K1t+P1} follows from the fact that for every $t \geq 1$, {\sc Graph Isomorphism} is solvable in polynomial-time on $(\overline{K_{1,t}+P_1},K_{1,t}+\nobreak P_1)$-free graphs~\cite[Theorem~4.2 with $b=b'=t$ and $c=c'=1$]{KS12}.
Statement~\ref{thm:gi-classification2:poly:P3+tP1-coP3+tP1} follows from the fact that for every $t \geq 1$, {\sc Graph Isomorphism} is solvable in polynomial-time on $(\overline{tP_1+P_3},tP_1+\nobreak P_3)$-free graphs~\cite[Theorem~4.2 and~4.3 with $b=b'=2$ and $c=c'=t$]{KS12}.
Statement~\ref{thm:gi-classification2:poly:Kt} follows from the fact that for every $t \geq 1$, {\sc Graph Isomorphism} is solvable in polynomial-time on
$(K_t,2K_{1,t})$-free graphs~\cite[Corollary~3 with $s=t$]{Sc17} (see also Lemma~\ref{lem:K_t-K_1t-free-gi-poly}),
$(K_t,K_{1,t}^+)$-free graphs~\cite[Theorem~16 with $b=t-1$ and $s=t$]{Sc17}
and $(K_t,P_5)$-free graphs~\cite[Theorem~14]{Sc17}.
Statement~\ref{thm:gi-classification2:poly:paw} follows from the fact that $(\paw,H)$-free graphs have bounded clique-width if $H \in \{P_2+\nobreak P_4, P_6, S_{1,2,2}\}$ (Theorem~\ref{thm:classification2}.\ref{thm:classification2:bdd:P_1+P_3}) combined with Theorem~\ref{thm:gi-poly-bdd-cw}, along with the fact that for every $t \geq 1$ {\sc Graph Isomorphism} is solvable in polynomial-time on $(K_3,K_{1,t}^{++}+\nobreak P_1)$-free graphs~\cite[Theorem~15 with $b=t-1$]{Sc17} and this class is equivalent to the class of $(\paw,K_{1,t}^{++}+\nobreak P_1)$-free graphs.
Statement~\ref{thm:gi-classification2:poly:diamond} follows from the fact that $(\diamondgraph,P_1+\nobreak 2P_2)$-free graphs have bounded clique-width (Theorem~\ref{thm:classification2}.\ref{thm:classification2:bdd:2P_1+P_2}) combined with Theorem~\ref{thm:gi-poly-bdd-cw}.
Similarly, Statement~\ref{thm:gi-classification2:poly:gem} follows from the fact that $(\gem,H)$-free graphs have bounded clique-width if $H \in \{P_1+\nobreak P_4,P_5\}$ (Theorem~\ref{thm:classification2}.\ref{thm:classification2:bdd:P_1+P_4}) combined with Theorem~\ref{thm:gi-poly-bdd-cw}.
Statement~\ref{thm:gi-classification2:poly:crossed-house} follows from the fact that {\sc Graph Isomorphism} is solvable in polynomial-time on $(\overline{2P_1+P_3},P_2+\nobreak P_3)$-free graphs (Theorem~\ref{thm:co2P1+P3-2P1+P3-easy})
and $(\overline{2P_1+P_3},P_5)$-free graphs (Theorem~\ref{thm:coP5-2P1+P3-easy}).

\begin{sloppypar}
Next, we consider the \GI-complete cases.
Statement~\ref{thm:gi-classification2:hard:not-path-star-forest} is~\cite[Lemma~2]{KS12}.
Statement~\ref{thm:gi-classification2:hard:not-co-path-star-forest} follows from Statement~\ref{thm:gi-classification2:hard:not-path-star-forest} since the class of $(\overline{H_1},\overline{H_2})$-free graphs is equivalent to the class of $(H_1,H_2)$-free graphs.
Statement~\ref{thm:gi-classification2:hard:K3} follows from the fact that {\sc Graph Isomorphism} is \GI-complete on $H$-free bipartite graphs if $H \in \{2P_1+\nobreak 2P_2,2P_1+\nobreak P_4,3P_2\}$~\cite[Lemma~5]{KS12} or $H=P_1+\nobreak 2P_3$~\cite[Theorem~6]{Sc17}.
Statement~\ref{thm:gi-classification2:hard:K4} follows from the fact that {\sc Graph Isomorphism} is \GI-complete on $(K_4,H)$-free graphs if $H \in \{K_{1,4}^{++}, P_1+\nobreak 2P_2\}$~\cite[Theorem~5]{Sc17} or $H=2P_1+\nobreak P_4$~\cite[Theorem~3]{KS12}.
Statement~\ref{thm:gi-classification2:hard:K5} follows from the fact that {\sc Graph Isomorphism} is \GI-complete on $(K_5,K_{1,3}^{++})$-free graphs~\cite[Theorem~7]{Sc17}.
Statement~\ref{thm:gi-classification2:hard:2P2} follows from the fact that {\sc Graph Isomorphism} is \GI-complete on $(C_4,C_5,3P_1+\nobreak P_2,2P_2)$-free graphs~\cite[Lemma~6 with $i=2$]{KS12} and for $(C_4,\diamondgraph,K_{1,3})$-free graphs~\cite[Lemma~9]{KS12}.
Statement~\ref{thm:gi-classification2:hard:2P1+P2} follows from the fact that {\sc Graph Isomorphism} is \GI-complete on $(C_4,\diamondgraph,K_{1,3})$-free graphs~\cite[Lemma~9]{KS12} and on $(\diamondgraph,H)$-free graphs if~$H$ is $P_2+\nobreak P_4$~\cite[Theorem~3]{DHP0}, $2P_3$ (Theorem~\ref{thm:diamond-2P3-hard}) or~$P_6$ (Theorem~\ref{thm:diamond-P6-hard}).
Statement~\ref{thm:gi-classification2:hard:gem} follows from the fact that {\sc Graph Isomorphism} is \GI-complete on $(\overline{P_1+P_4},P_1+\nobreak 2P_2)$-free graphs (Theorem~\ref{thm:gem-P1+2P2-hard}).\qed
\end{sloppypar}
\end{proof}

\begin{sloppypar}
\begin{oproblem}\label{oprob:gi}
What is the complexity of {\sc Graph Isomorphism} on $(H_1,H_2)$-free graphs in the following cases?
\begin{enumerate}[(i)]
\renewcommand{\theenumi}{(\roman{enumi})}
\renewcommand{\labelenumi}{(\roman{enumi})}
\item \label{oprob:gi:3P_1}$H_1=K_3$ and $H_2 \in \{P_7,S_{1,2,3}\}$
\item \label{oprob:gi:4P_1}$H_1=K_4$ and $H_2=S_{1,1,3}$
\item \label{oprob:gi:diamond}$H_1=\diamondgraph$ and $H_2 \in \{P_1+\nobreak P_2+\nobreak P_3,P_1+\nobreak P_5\}$
\item \label{oprob:gi:gem}$H_1=\gem$ and $H_2=P_2+\nobreak P_3$
\end{enumerate}
\end{oproblem}
\end{sloppypar}

Note that all of the classes of $(H_1,H_2)$-free graphs in Open Problem~\ref{oprob:gi} are incomparable.
The following theorem states that Open Problem~\ref{oprob:gi} lists all open cases.

\begin{theorem}\label{thm:opencases}
Let~${\cal G}$ be a class of graphs defined by two forbidden induced subgraphs.
Then~${\cal G}$ is not equivalent to any of the classes listed in Theorem~\ref{thm:gi-classification2} if and only if it is equivalent to one of the six cases listed in Open Problem~\ref{oprob:gi}.
\end{theorem}

\begin{proof}
It is easy to verify that none of the classes in Open Problem~\ref{oprob:gi} are equivalent to any of the classes in Theorem~\ref{thm:gi-classification2}.

Let $H_1,H_2$ be graphs and let~${\cal G}$ be the class of $(H_1,H_2)$-free graphs.
Suppose~${\cal G}$ is not equivalent to any class for which the complexity of {\sc Graph Isomorphism} is implied by Theorem~\ref{thm:gi-classification2}.
We will show that ${\cal G}$ is equivalent one of the classes in Open Problem~\ref{oprob:gi}.
By Theorem~\ref{thm:gi-classification2}.\ref{thm:gi-classification2:poly:P4}, we may assume that $H_1,H_2 \not\ssi P_4$.
Since $P_4=\overline{P_4}$, this means that none of $H_1,\overline{H_1},H_2,\overline{H_2}$ are induced subgraphs of~$P_4$.

By Theorem~\ref{thm:gi-classification2}.\ref{thm:gi-classification2:hard:not-path-star-forest}, at least one of~$H_1$ and~$H_2$ must be a path star forest.
By Theorem~\ref{thm:gi-classification2}.\ref{thm:gi-classification2:hard:not-co-path-star-forest}, at least one of~$\overline{H_1}$ and~$\overline{H_2}$ must be a path star forest.
Suppose, for contradiction, that both~$H_1$ and~$\overline{H_1}$ are path star forests.
Let~$n$ be the number of vertices in~$H_1$.
Then~$H_1$ and~$\overline{H_1}$ each contain at most $n-1$ edges.
Since~$H_1$ and~$\overline{H_1}$ together have~$\binom{n}{2}$ edges, it follows that $\binom{n}{2}\leq 2(n-1)$ and so $n \leq 4$.
It is easy to verify that if~$F$ is a forest on at most four vertices and~$\overline{F}$ is also a forest, then~$F$ is an induced subgraph of~$P_4$.
Therefore~$H_1$ is an induced subgraph of~$P_4$, a contradiction.
By symmetry, we may therefore assume that~$\overline{H_1}$ and~$H_2$ are path star forests, but~$H_1$ and~$\overline{H_2}$ are not.

Also note that by definition of equivalence, the theorem is symmetric in~$\overline{H_1}$ and~$H_2$.
We will consider a number of cases, depending on the possibilities for~$H_1$.
First, we consider the cases when $H_1=K_s$ for some $s \geq 1$.
Since $H \not\ssi P_4$, we may assume that $s \geq 3$.

\thmcase{\label{case:K3}$H_1=K_3$.}
By Theorem~\ref{thm:gi-classification2}.\ref{thm:gi-classification2:hard:K3}, we may assume that~$H_2$ is $(2P_1+\nobreak 2P_2,P_1+\nobreak 2P_3,2P_1+\nobreak P_4,3P_2)$-free.

First consider the case when~$H_2$ is a $P_4$-free path star forest, or equivalently when~$H_2$ is a disjoint union of stars.
Since~$H_2$ is $3P_2$-free, it has at most two non-trivial components.
If~$H_2$ has at most one non-trivial component, then it is an induced subgraph of $K_{1,t} +\nobreak  tP_1 \ssi 2K_{1,t}$ for some $t \geq 1$ and so Theorem~\ref{thm:gi-classification2}.\ref{thm:gi-classification2:poly:Kt} applies.
If~$H_2$ has two non-trivial components, at least one of which is isomorphic to~$P_2$, then~$H_2$ has at most three components since it is $(2P_1+\nobreak 2P_2)$-free, and so~$H_2$ is an induced subgraph of $K_{1,t}+\nobreak P_2 +\nobreak P_1 \ssi K_{1,t+1}^{++} +\nobreak P_1$ for some $t \geq 1$ and so Theorem~\ref{thm:gi-classification2}.\ref{thm:gi-classification2:poly:paw} applies.
If~$H_2$ has two non-trivial components, neither of which is isomorphic to~$P_2$, then both of these components contain an induced~$P_3$.
In this case, since~$H_2$ is $(P_1+\nobreak 2P_3)$-free, $H_2$ contains exactly two components, so it is an induced subgraph of~$2K_{1,t}$ for some $t \geq 1$ and thus Theorem~\ref{thm:gi-classification2}.\ref{thm:gi-classification2:poly:Kt} applies.
Therefore we may assume that~$H_2$ contains an induced~$P_4$.

Let~$C$ be the component of~$H_2$ that contains this induced~$P_4$.
Since~$H_2$ is $(2P_1+\nobreak P_4)$-free, $H_2$ contains at most one component apart from~$C$.
Furthermore, if it does contain a second component, then that component must isomorphic to~$P_1$ or~$P_2$.
In other words, $H_2$ is isomorphic to~$C$, $C+\nobreak P_1$ or~$C+\nobreak P_2$.

If $H_2=C+\nobreak P_2$, then since~$H_2$ is $(2P_1+\nobreak 2P_2,3P_2)$-free, it follows that~$C$ is a $(2P_1+\nobreak P_2,2P_2)$-free tree that contains an induced~$P_4$.
Since~$C$ is $2P_2$-free, the end-vertices of the induced~$P_4$ cannot have a neighbour outside the~$P_4$ and since it is $(2P_1+\nobreak P_2)$-free, the two internal vertices of the~$P_4$ cannot have a neighbour outside the~$P_4$.
Therefore $H_2=P_2+\nobreak P_4$ and so Theorem~\ref{thm:gi-classification2}.\ref{thm:gi-classification2:poly:paw} applies.
We may therefore assume that $H_2 \neq C+\nobreak P_2$.

Suppose that $H_2=C+\nobreak P_1$.
Since~$H_2$ is $(2P_1+\nobreak 2P_2,2P_1+\nobreak P_4)$-free, it follows that~$C$ is $(P_1+\nobreak 2P_2,P_1+\nobreak P_4)$-free.
Since~$C$ is $(P_1+\nobreak P_4)$-free, the~$P_4$ dominates~$C$ and at most one of the end-vertices of the~$P_4$ has a neighbour outside this~$P_4$.
Since~$H_2$ is a path star forest, it has at most one vertex of degree greater than~$2$.
Therefore, since the~$P_4$ dominates~$C$, it follows that~$C$ is obtained from~$P_4$ or~$P_5$ by attaching a (possibly empty) set of pendant edges to one of its internal vertices.
Since~$C$ is $(P_1+\nobreak 2P_2)$-free, it cannot be obtained from~$P_5$ by adding a non-zero number of pendant vertices adjacent to the central vertex.
Therefore~$C$ is obtained from~$P_4$ or~$P_5$ by adding~$t$ pendant vertices to a vertex adjacent to an end-vertex of this path for some $t \geq 0$.
It follows that $H_2=K_{1,t+2}^++\nobreak P_1$ or $H_2=K_{1,t+2}^{++}+\nobreak P_1$, respectively and so Theorem~\ref{thm:gi-classification2}.\ref{thm:gi-classification2:poly:paw} applies.
We may therefore assume that $H_2 \neq C+\nobreak P_1$.

Finally, suppose that $H_2=C$, in which case~$H_2$ is connected.
Then it is obtained from~$K_{1,t}$ for some $t \geq 2$ by subdividing edges.
If~$t=2$, then~$H_2$ is isomorphic to~$P_k$ for some~$k \geq 4$, and $k \leq 7$ since~$H_2$ is $3P_2$-free.
If $k=7$, then Open Problem~\ref{oprob:gi}.\ref{oprob:gi:3P_1} applies, and if $k \leq 6$, then Theorem~\ref{thm:gi-classification2}.\ref{thm:gi-classification2:poly:paw} applies.
We may therefore assume that $t \geq 3$.
Since~$H_2$ is $(2P_1+\nobreak P_4)$-free, each edge of this~$K_{1,t}$ can be subdivided at most twice.
If $t \geq 4$, then at most one of the edges of the~$K_{1,t}$ can be subdivided since~$H_2$ is $(2P_1+\nobreak 2P_2)$-free and so in this case $H_2 \ssi K_{1,t}^{++}$ and Theorem~\ref{thm:gi-classification2}.\ref{thm:gi-classification2:poly:paw} applies.
We may therefore assume that~$t=3$, so $H_2=S_{i,j,k}$ for some $1 \leq i \leq j \leq k$.
Now $k \geq 2$ since~$H_2$ contains an induced~$P_4$ and $k \leq 3$ since each edge of the~$K_{1,t}$ is subdivided at most twice.
If $k=2$ or $j=1$, then $H_2 \ssi S_{1,2,2}$ or $H_2 \ssi S_{1,1,3}=K_{1,t}^{++}$ for $t=3$, respectively and Theorem~\ref{thm:gi-classification2}.\ref{thm:gi-classification2:poly:paw} applies, so we may assume $j=2$ and $k=3$.
Therefore $H_2=S_{1,2,3}$ and Open Problem~\ref{oprob:gi}.\ref{oprob:gi:3P_1} applies.
This completes the proof for Case~\ref{case:K3}.

\thmcase{\label{case:K4}$H_1=K_s$ for some $s \geq 4$.}
By Theorem~\ref{thm:gi-classification2}.\ref{thm:gi-classification2:hard:K4}, we may assume that~$H_2$ is $(K_{1,4}^{++},P_1+\nobreak 2P_2,P_1+\nobreak P_4)$-free.
Since~$H_2$ is $(P_1+\nobreak 2P_2)$-free, if it contains two non-trivial components, then it contains no other components.
Thus if every component of~$H_2$ is a star, then either~$H_2$ contains only two components, or~$H_2$ contains at most one non-trivial component and all other components are trivial.
In the first case $H_2 \ssi 2K_{1,t}$ for some $t \geq 1$ and in the second case $H_2 \ssi K_{1,t}+\nobreak tP_1 \ssi 2K_{1,t}$ for some $t \geq 1$.
Therefore, if every component of~$H_2$ is a star, then Theorem~\ref{thm:gi-classification2}.\ref{thm:gi-classification2:poly:Kt} applies.
We may therefore assume that~$H_2$ is not a disjoint union of stars.
Since~$H_2$ is a forest, this implies that~$P_4$ is an induced subgraph of~$H_2$.
Since~$H_2$ is $(P_1+\nobreak P_4)$-free, this~$P_4$ must dominate~$H_2$ and~$H_2$ cannot be isomorphic to~$P_k$ for $k \geq 6$.
In particular, note that this implies that~$H_2$ is connected.
If~$H_2$ has maximum degree at most~$2$, then $H_2 \ssi P_5$ and Theorem~\ref{thm:gi-classification2}.\ref{thm:gi-classification2:poly:Kt} applies.
We may therefore assume that~$H_2$ is obtained by subdividing edges of~$K_{1,t}$ for some $t \geq 3$.
Since~$H_2$ is $(P_1+\nobreak 2P_2)$-free, at most one edge of the~$K_{1,t}$ can be subdivided.
Since~$H_2$ is $(P_1+\nobreak P_4)$-free, any edge of the~$K_{1,t}$ can be subdivided at most twice, and so $H_2 \ssi K_{1,t}^{++}$.
If~$H_2$ is obtained from~$K_{1,t}$ by subdividing an edge at most once, then $H_2 \ssi K_{1,t}^+$ and Theorem~\ref{thm:gi-classification2}.\ref{thm:gi-classification2:poly:Kt} applies, so we may assume that $H_2=K_{1,t}^{++}$.
Since~$H_2$ is $K_{1,4}^{++}$-free, it follows that $t=3$ and so $H_2=K_{1,3}^{++}=S_{1,1,3}$.
Now Open Problem~\ref{oprob:gi}.\ref{oprob:gi:4P_1} or Theorem~\ref{thm:gi-classification2}.\ref{thm:gi-classification2:hard:K5} applies if $s=4$ or $s \geq 5$, respectively.
This completes the proof for Case~\ref{case:K4}.

\medskip
\noindent
For the remainder of the proof we may therefore assume that Cases~\ref{case:K3} and~\ref{case:K4} do not apply, so~$H_1$ is not a complete graph.
By symmetry between~$\overline{H_1}$ and~$H_2$, we may thus assume that both these graphs contain an edge.
Furthermore, by definition of equivalence, if~$\overline{H_1}$ or~$H_2$ is isomorphic to $P_1+\nobreak P_3=\overline{\paw}$, then we can replace the graph in question by~$3P_1=\overline{K_3}$.
Thus Case~\ref{case:K3} completes the proof if~$\overline{H_1}$ or~$H_2$ is either~$3P_1$ or~$P_1+\nobreak P_3$.
Every induced subgraph of~$P_1+\nobreak P_3$, other than~$3P_1$ and~$P_1+\nobreak P_3$, is an induced subgraph of~$P_4$, and we assumed that neither~$\overline{H_1}$ nor~$H_2$ is an induced subgraph of~$P_4$.
In the remainder of the proof we may therefore assume that neither~$\overline{H_1}$ nor~$H_2$ is an induced subgraph of~$P_1+\nobreak P_3$ or of~$P_4$.

\thmcase{\label{case:not-linear-forest}$\overline{H_1}$ not a linear forest.}
In this case~$\overline{H_1}$ contains a vertex of degree at least~$3$, so it contains an induced~$K_{1,3}$.
Note that $\overline{C_4}=2P_2$ and $\overline{\diamondgraph}=2P_1+\nobreak P_2$.
Therefore, by Theorems~\ref{thm:gi-classification2}.\ref{thm:gi-classification2:hard:2P2} and~\ref{thm:gi-classification2}.\ref{thm:gi-classification2:hard:2P1+P2}, respectively, we may assume that~$H_2$ is $2P_2$-free and $(2P_1+\nobreak P_2)$-free.
Since~$H_2$ is $2P_2$-free, it has at most one non-trivial component.
Furthermore, every non-trivial component of~$H_2$ must be a $2P_2$-free path star, so it must be isomorphic to~$K_{1,k}$ or~$K_{1,k}^+$ for some~$k \geq 1$.
Recall that we may assume that~$H_2$ contains at least one non-trivial component, otherwise we reduce to Case~\ref{case:K3} or~\ref{case:K4}.
Therefore, since~$H_2$ is $(2P_1+\nobreak P_2)$-free, it can have at most one trivial component and we conclude that $H_2 \in \{K_{1,k}, K_{1,k}^+,K_{1,k}+\nobreak P_1,K_{1,k}^++\nobreak P_1\}$ for some $k \geq 1$.
If~$k \leq 2$, then either~$H_2$ is an induced subgraph of~$P_1+\nobreak P_3$ or~$P_4$, or $H_2 = K_{1,2}^++\nobreak P_1=P_1+\nobreak P_4$, in which case~$H_2$ contains an induced~$2P_1+\nobreak P_2$, a contradiction.
We may therefore assume that $k\geq 3$, in which case $H_2 \notin \{K_{1,k}^+,K_{1,k}^++\nobreak P_1\}$ since~$H_2$ is $(2P_1+\nobreak P_2)$-free.
Thus $H_2\in \{K_{1,k},K_{1,k}+\nobreak P_1\}$ for some $k \geq 3$.
In particular, this implies  $K_{1,3} \ssi H_2$, so by the same argument with~$H_2$ taking the part of~$\overline{H_1}$, we may assume that $\overline{H_1} \in \{K_{1,t}, K_{1,t}+\nobreak P_1\}$ for some $t \geq 3$.
Therefore Theorem~\ref{thm:gi-classification2}.\ref{thm:gi-classification2:poly:K1t+P1-co-K1t+P1} applies.
This completes the proof for Case~\ref{case:not-linear-forest}.

\medskip
\noindent
For the remainder of the proof we may therefore assume that Case~\ref{case:not-linear-forest} does not apply.
By symmetry between~$\overline{H_1}$ and~$H_2$, we may thus assume that both these graphs are linear forests.

\thmcase{\label{case:contains-P5}$\overline{H_1}$ contains~$P_5$ as an induced subgraph.}
Recall that we may assume~$H_2$ contains a non-trivial component, otherwise we reduce to Case~\ref{case:K3} or~\ref{case:K4}.
Note that $P_5 \si 2P_2=\overline{C_4}$.
Therefore, by Theorem~\ref{thm:gi-classification2}.\ref{thm:gi-classification2:hard:2P2}, we may assume that~$H_2$ is $(3P_1+\nobreak P_2,2P_2)$-free.
Since~$H_2$ is $2P_2$-free, it has exactly one non-trivial component, which must be isomorphic to~$P_t$, for some $2 \leq t \leq 4$.
Since~$H_2$ is not an induced subgraph of~$P_4$, it follows that~$H_2$ is isomorphic to $sP_1+\nobreak P_t$ for some $s \geq 1$ and $t \in \{2,3,4\}$.
Since~$H_2$ is $(3P_1+\nobreak P_2)$-free, it follows that $s \leq 2$ and if $t=4$, then $s=1$.
Since~$H_2$ is not an induced subgraph of~$P_1+\nobreak P_3$, if $s=1$, then $t=4$.
Therefore $H_2 \in \{2P_1+\nobreak P_2,2P_1+\nobreak P_3,P_1+\nobreak P_4\}$.
First consider the case when $H_2=P_1+\nobreak P_4$.
By Theorems~\ref{thm:gi-classification2}.\ref{thm:gi-classification2:hard:2P1+P2} and~\ref{thm:gi-classification2}.\ref{thm:gi-classification2:hard:gem}, respectively, we may assume that~$\overline{H_1}$ is $P_6$-free and $(P_1+\nobreak 2P_2)$-free.
Since~$\overline{H_1}$ contains~$P_5$ as an induced subgraph, but is $(P_1+\nobreak 2P_2)$-free, it follows that~$\overline{H_1}$ is connected.
Since~$\overline{H_1}$ is $P_6$-free, it follows that $\overline{H_1}=P_5$, and so Theorem~\ref{thm:gi-classification2}.\ref{thm:gi-classification2:poly:gem} applies.
This completes the case when $H_2=P_1+\nobreak P_4$ and so we may assume that $H_2 \in \{2P_1+\nobreak P_2,2P_1+\nobreak P_3\}$.
By Theorems~\ref{thm:gi-classification2}.\ref{thm:gi-classification2:hard:K3} and~\ref{thm:gi-classification2}.\ref{thm:gi-classification2:hard:2P1+P2}, respectively, we may assume that~$\overline{H_1}$ is $(2P_1+\nobreak 2P_2)$-free and $(P_2+\nobreak P_4,P_6)$-free.
Since~$\overline{H_1}$ is a $P_6$-free linear forest that contains~$P_5$ as an induced subgraph, it follows that~$\overline{H_1}$ contains a component isomorphic to~$P_5$.
Since~$\overline{H_1}$ is $(2P_1+\nobreak 2P_2,P_2+\nobreak P_4)$-free, it follows that~$\overline{H_1}$ contains at most one vertex outside this component, so $\overline{H_1} \in \{P_5,P_1+\nobreak P_5\}$.
If $\overline{H_1}=P_5$, then Theorem~\ref{thm:gi-classification2}.\ref{thm:gi-classification2:poly:gem} applies if $H_2=2P_1+\nobreak P_2 \ssi P_1+\nobreak P_4$ and Theorem~\ref{thm:gi-classification2}.\ref{thm:gi-classification2:poly:crossed-house} applies of $H_2=2P_1+\nobreak P_3$.
If $\overline{H_1}=P_1+\nobreak P_5$, then Open Problem~\ref{oprob:gi}.\ref{oprob:gi:diamond} applies if $H_2=2P_1+\nobreak P_2$ and Theorem~\ref{thm:gi-classification2}.\ref{thm:gi-classification2:hard:K4} applies if $H_2=2P_1+\nobreak P_3 \si 4P_1$.
This completes the proof for Case~\ref{case:contains-P5}.

\thmcase{\label{case:contains-P4}$\overline{H_1}$ contains~$P_4$ as an induced subgraph.}
By Case~\ref{case:not-linear-forest} we may assume that~$\overline{H_1}$ and~$H_2$ are linear forests and by Cases~\ref{case:K3} and~\ref{case:K4}, we may assume they each contain at least one non-trivial component.
By Case~\ref{case:contains-P5}, we may assume that~$\overline{H_1}$ is $P_5$-free, so it contains a component isomorphic to~$P_4$.
Since~$\overline{H_1}$ is not an induced subgraph of~$P_4$, it follows that~$\overline{H_1}$ contains at least one other component.
First consider the case when~$\overline{H_1}$ contains a non-trivial component apart from this~$P_4$, so $P_2+\nobreak P_4 \ssi \overline{H_1}$.
In this case Theorems~\ref{thm:gi-classification2}.\ref{thm:gi-classification2:hard:2P2} and~\ref{thm:gi-classification2}.\ref{thm:gi-classification2:hard:2P1+P2}, respectively, imply that~$H_2$ is $2P_2$-free and $(2P_1+\nobreak P_2)$-free.
Since~$H_2$ is $2P_2$-free, it has one non-trivial component, which must be isomorphic to~$P_t$ for some $2 \leq t \leq 4$.
Since~$H_2$ is $(2P_1+\nobreak P_2)$-free, it follows that~$H_2$ is an induced subgraph of~$P_1+\nobreak P_3$ or~$P_4$, a contradiction.
We conclude that~$\overline{H_1}$ cannot contain any non-trivial components apart from the~$P_4$ and so $\overline{H_1}=tP_1+\nobreak P_4$ for some $t \geq 1$.
If $t \geq 2$, then by Theorem~\ref{thm:gi-classification2}.\ref{thm:gi-classification2:hard:K3} we may assume that~$H_2$ is $3P_1$-free.
Since~$H_2$ is a linear forest that is not an induced subgraph of~$P_4$, this implies that $H_2=2P_2$, in which case Theorem~\ref{thm:gi-classification2}.\ref{thm:gi-classification2:hard:2P2} applies.
We may therefore assume that $t=1$ and so $\overline{H_1}=P_1+\nobreak P_4$.
By symmetry, if~$H_2$ contains a~$P_4$ as an induced subgraph, then we may assume $H_2=P_1+\nobreak P_4$, in which case Theorem~\ref{thm:gi-classification2}.\ref{thm:gi-classification2:poly:gem} applies.
We may therefore assume that~$H_2$ is $P_4$-free, so every component of~$H_2$ is isomorphic to~$P_1$, $P_2$ or~$P_3$.
By Theorems~\ref{thm:gi-classification2}.\ref{thm:gi-classification2:hard:K4}, \ref{thm:gi-classification2}.\ref{thm:gi-classification2:hard:2P1+P2} and~\ref{thm:gi-classification2}.\ref{thm:gi-classification2:hard:gem}, respectively, we may assume that~$H_2$ is $4P_1$-free, $2P_3$-free and $(P_1+\nobreak 2P_2)$-free.
Since~$H_2$ is $2P_3$-free, it contains at most one component isomorphic to~$P_3$.
Since~$H_2$ is $(P_1+\nobreak 2P_2)$-free, if it contains two non-trivial components, then it contains no other components.
In this case $H_2=2P_2\ssi P_5$ or $H_2=P_2+\nobreak P_3$, in which case Theorem~\ref{thm:gi-classification2}.\ref{thm:gi-classification2:poly:gem} or Open Problem~\ref{oprob:gi}.\ref{oprob:gi:gem}, respectively, applies.
We may therefore assume that~$H_2$ contains exactly one non-trivial component.
Since~$H_2$ is $4P_1$-free, but not an induced subgraph of $P_1+\nobreak P_3$, it follows that $H_2=2P_1+\nobreak P_2\ssi P_1+\nobreak P_4$ and so Theorem~\ref{thm:gi-classification2}.\ref{thm:gi-classification2:poly:gem} applies.
This completes the proof for Case~\ref{case:contains-P4}.

\thmcase{\label{case:contains-2P2}$\overline{H_1}$ contains~$2P_2$ as an induced subgraph.}
We may assume that~$H_2$ contains a non-trivial component, otherwise we reduce to Case~\ref{case:K3} or~\ref{case:K4}.
Furthermore, we may assume that~$H_2$ is a $P_4$-free linear forest, otherwise we reduce to Case~\ref{case:not-linear-forest} or~\ref{case:contains-P4}.
By Theorem~\ref{thm:gi-classification2}.\ref{thm:gi-classification2:hard:2P2}, we may assume that~$H_2$ is $(3P_1+\nobreak P_2,\allowbreak 2P_2)$-free.
Since~$H_2$ is $2P_2$-free, but contains at least one non-trivial component, it follows that~$H_2$ contains exactly one non-trivial component.
Furthermore, since~$H_2$ is $P_4$-free, this non-trivial component is isomorphic to either~$P_2$ or~$P_3$.
Since~$H_2$ is $(3P_1+\nobreak P_2)$-free, but not an induced subgraph of~$P_1+\nobreak P_3$, it follows that $H_2 \in \{2P_1+\nobreak P_2,2P_1+\nobreak P_3\}$.
By Theorems~\ref{thm:gi-classification2}.\ref{thm:gi-classification2:hard:K3} and~\ref{thm:gi-classification2}.\ref{thm:gi-classification2:hard:2P1+P2}, respectively, we may assume that~$\overline{H_1}$ is $(2P_1+\nobreak 2P_2,3P_2)$-free and $2P_3$-free.
Since~$\overline{H_1}$ is $3P_2$-free, it has at most two non-trivial components and since it contains~$2P_2$ as an induced subgraph, it must contain at least two non-trivial components.
Since~$\overline{H_1}$ is $2P_3$-free, its non-trivial components must either both be isomorphic to~$P_2$, or one of these components is isomorphic to~$P_2$ and the other to~$P_3$.
We may assume that~$\overline{H_1}$ has another component, otherwise $\overline{H_1} \ssi P_2+\nobreak P_3$, in which case Theorem~\ref{thm:gi-classification2}.\ref{thm:gi-classification2:poly:crossed-house} applies.
Since~$\overline{H_1}$ is $(2P_1+\nobreak 2P_2)$-free, it has at most one other component, which must be trivial.
We conclude that $\overline{H_1} \in \{P_1+\nobreak 2P_2, P_1+\nobreak P_2 + \nobreak P_3\}$.
By Theorem~\ref{thm:gi-classification2}.\ref{thm:gi-classification2:hard:K4}, we may assume that~$H_2$ is $4P_1$-free, so $H_2=2P_1+\nobreak P_2$.
Theorem~\ref{thm:gi-classification2}.\ref{thm:gi-classification2:poly:diamond} or Open Problem~\ref{oprob:gi}.\ref{oprob:gi:diamond} applies if $\overline{H_1}=2P_1+\nobreak P_2$ or $P_1+\nobreak P_2 + \nobreak P_3$, respectively.
This completes the proof for Case~\ref{case:contains-2P2}.

\medskip
\noindent
By Case~\ref{case:not-linear-forest} we may assume that~$\overline{H_1}$ and~$H_2$ are both linear forests.
By Cases~\ref{case:contains-P4} and~\ref{case:contains-2P2}, we may assume that they are both $(2P_2,P_4)$-free.
Since, $\overline{H_1}$ and~$H_2$ are $2P_2$-free, they each contain at most one non-trivial component.
Since they are $P_4$-free, any such non-trivial component must be isomorphic to~$P_2$ or~$P_3$.
Therefore~$\overline{H_1}$ and~$H_2$ must both be induced subgraphs of~$tP_1+\nobreak P_3$ for some $t \geq 1$.
In this case Theorem~\ref{thm:gi-classification2}.\ref{thm:gi-classification2:poly:P3+tP1-coP3+tP1} applies.
This completes the proof.\qed
\end{proof}

\section{Conclusions}\label{s-con}

By combining known and new results, we determined the complexity of {\sc Graph Isomorphism} in terms of polynomial-time solvability and \GI-completeness for $(H_1,H_2)$-free graphs for all but six pairs $(H_1,H_2)$.
This also led to a new class of $(H_1,H_2)$-free graphs whose clique-width is unbounded.
In particular, we developed a technique for showing polynomial-time solvability of {\sc Graph Isomorphism} for $(\overline{2P_1+P_3},H)$-free graphs, which we illustrated for the $H=P_2+\nobreak P_3$ and $H=P_5$ cases, thus completing the classification for $(\overline{2P_1+P_3},H)$-free graphs.
To obtain full dichotomies for the complexity of {\sc Graph Isomorphism} and the (un)boundedness of clique-width on $(H_1,H_2)$-free graphs, we need to solve the six remaining open cases for {\sc Graph Isomorphism} (see Open Problem~\ref{oprob:gi}) and five open cases for boundedness of clique-width (see Open Problem~\ref{oprob:twographs}).
We leave this as future work, but note that new techniques will be required to deal with these cases.

\bibliography{mybib}

\end{document}